\documentclass[11pt]{article}

\usepackage[margin=1in]{geometry}
\usepackage{indentfirst}
\usepackage{times}
\usepackage{amsmath}
\usepackage{bm}
\usepackage{amssymb}
\usepackage{amsthm}
\usepackage{mathtools}
\usepackage{mleftright}
\usepackage{thmtools}
\usepackage{comment}
\usepackage{algorithm}
\usepackage[noend]{algpseudocode}
\usepackage[numbers]{natbib}
\usepackage[colorlinks=true,linkcolor=blue,citecolor=blue,urlcolor=blue]{hyperref}
\usepackage[nameinlink,capitalize]{cleveref}
\usepackage{enumitem}
\usepackage{tikz}
\usetikzlibrary{arrows.meta,calc,positioning}
\newtheorem{theorem}{Theorem}[section]
\newtheorem{lemma}[theorem]{Lemma}
\newtheorem{claim}[theorem]{Claim}

\newtheorem{definition}[theorem]{Definition}

\RenewCommandCopy{\theHtheorem}{\thetheorem}

\DeclareMathOperator*{\Prob}{Pr}

\newcommand{\bool}{\mleft\{0,1\mright\}}

\newcommand{\poly}{\mathsf{poly}}

\newcommand{\PSAMP}{\mathsf{PSAMP}}
\newcommand{\supp}{\mathsf{Support}}

\newcommand{\eqdef}{\triangleq}

\renewcommand{\P}{\mathsf{P}}
\newcommand{\NP}{\mathsf{NP}}
\newcommand{\coNP}{\mathsf{coNP}}
\newcommand{\PH}{\mathsf{PH}}
\newcommand{\BPP}{\mathsf{BPP}}
\newcommand{\Avg}{\mathsf{Avg}}
\newcommand{\Dist}{\mathsf{Dist}}
\newcommand{\BPTIME}{\mathsf{BPTIME}}
\newcommand{\prBPP}{\mathsf{prBPP}}
\newcommand{\promise}{\mathsf{promise}\text{-}}

\newcommand{\DistPH}{\Dist\PH}

\newcommand{\AM}{\mathsf{AM}}

\newcommand{\SAT}{\mathsf{SAT}}

\newcommand{\K}{\mathsf{K}}
\newcommand{\Kt}{\mathsf{Kt}}
\newcommand{\pK}{\mathsf{pK}}

\newcommand{\nK}{\mathsf{nK}}

\newcommand{\pnK}{\mathsf{pnK}}

\newcommand{\lK}{\text{-}\mathsf{K}}
\newcommand{\lpK}{\text{-}\mathsf{pK}}
\newcommand{\lpnK}{\text{-}\mathsf{pnK}}
\newcommand{\MINKT}{\mathsf{MINKT}}

\newcommand{\MINpKT}{\mathsf{MINpKT}}

\newcommand{\MINlpKT}{\mathsf{MIN\ell\text{-}pKT}}

\newcommand{\GapMINpKT}{\mathsf{Gap}\textnormal{-}\allowbreak\MINpKT}
\newcommand{\MildGapMINpKT}{\mathsf{Mild}\textnormal{-}\allowbreak\mathsf{Gap}\textnormal{-}\allowbreak\MINpKT}
\newcommand{\GapMINKT}{\mathsf{Gap}\textnormal{-}\allowbreak\MINKT}
\newcommand{\numberthis}{\addtocounter{equation}{1}\tag{\theequation}}

\title{The Gap Between Average-Case and Worst-Case Complexity\vspace{0.4cm}}
\author{
Halley Goldberg\thanks{University of Warwick. \texttt{E-mail:~halley.goldberg@warwick.ac.uk}}
\and
Jinqiao Hu\thanks{University of Warwick. \texttt{E-mail:~jinqiao.hu@warwick.ac.uk}}
\and
Zhenjian Lu\thanks{University of Victoria. \texttt{E-mail:~zhenjianlu@uvic.ca}}
\and
Jingyi Lyu\thanks{University of Michigan. \texttt{E-mail:~jylyu@umich.edu}}
\and
Igor C. Oliveira\vspace{0.6cm}\thanks{University of Warwick. \texttt{E-mail:~igor.oliveira@warwick.ac.uk}}
}

\begin{document}
\maketitle

\vspace{-0.3cm}

\begin{abstract}
Understanding the average-case complexity of $\NP$ is a central goal in complexity theory and cryptography. A necessary milestone is to understand the corresponding average-case versus worst-case question for the polynomial hierarchy ($\PH$). In a celebrated result, Hirahara \citep{DBLP:conf/stoc/Hirahara21} proved that if every distributional problem in $\PH$ is easy on average, then every language in $\PH$ admits worst-case algorithms running in time $2^{O(n/\log n)}$. Remarkably, although the statement of Hirahara's theorem makes no reference to meta-complexity, all known proofs appear to use time-bounded Kolmogorov complexity and its associated meta-computational problems in an essential way.

We show that the reliance on techniques and perspectives from meta-complexity is not simply a coincidence. Building on previous work of Hirahara \citep{DBLP:conf/focs/Hirahara20} and on a recent result of Hu, Lu, and Oliveira \citep{DBLP:journals/eccc/HuL025}, we prove that relating the average-case and worst-case complexities of $\PH$ is in fact \emph{equivalent} to a purely meta-computational implication: efficient algorithms for a \emph{standard gap version} of $\MINpKT$ with $\PH$ oracle access imply efficient algorithms for a \emph{mild gap version}. Crucially, in these promise problems the standard gap separates a YES instance with time bound $t$ from a NO instance with time bound $\poly(t)$, whereas the mild gap keeps the YES threshold $t$ and moves the NO threshold to $t+\poly(n)$. 

Together with prior work, this equivalence points to the study of \emph{gap reductions} as a potentially promising route toward understanding the relation between average-case and worst-case complexity. We make progress by proving that the gap in running time can be eliminated in the average-case setting.  Finally, we prove that in the worst-case setting any efficient reduction between the two computational problems  must use nonrelativizing techniques.
\end{abstract}

\normalsize

\clearpage

\tableofcontents

\clearpage

\section{Introduction}

\subsection{Overview}

Understanding the average-case complexity of $\NP$ is a fundamental research direction in complexity theory. A central problem in this area asks whether worst-case hardness of $\NP$ can be transformed into average-case hardness of $\NP$ against efficient algorithms. This question is particularly relevant to the complexity-theoretic foundations of cryptography, where one would ideally like to base cryptographic primitives on worst-case intractability assumptions about  $\NP$ problems. We refer to the survey of Bogdanov and Trevisan \citep{DBLP:journals/fttcs/BogdanovT06} for background on average-case complexity and its connections to cryptography.

Formally, before resolving the average-case complexity of $\NP$, one must first understand the analogous question for the Polynomial Hierarchy ($\PH$).\footnote{In other words, the implication $\mathsf{DistPH} \subseteq \mathsf{AvgBPP} \Longrightarrow \mathsf{PH} \subseteq \mathsf{BPP}$ follows from the implication $\mathsf{DistNP} \subseteq \mathsf{AvgBPP} \Longrightarrow \mathsf{NP} \subseteq \mathsf{BPP}$.} In the language of Impagliazzo's five worlds \citep{DBLP:conf/coco/Impagliazzo95}, the desired statement is the exclusion of $\PH$ Heuristica:
\begin{equation}\label{eq:excl_PH_Heuristica}
    \Dist\PH \subseteq \Avg\BPP \quad \Longrightarrow \quad \PH \subseteq \BPP.
\end{equation}
Equivalently, if every problem in $\PH$ is easy on every polynomial-time samplable distribution, then every problem in $\PH$ is easy in the worst case.

A few years ago, a breakthrough result of Hirahara \citep{DBLP:conf/stoc/Hirahara21} (see also \citep{GKLO22}) showed that a weaker worst-case conclusion follows from the average-case easiness of $\PH$:
\begin{equation}\label{eq:Hir21}
       \Dist\PH \subseteq \Avg\BPP
    \quad \Longrightarrow \quad
    \PH \subseteq \BPTIME[2^{O(n/\log n)}].
\end{equation}
Hirahara's Theorem (\Cref{eq:Hir21}) offered the first significant progress on this problem after a long gap.  Note that the theorem itself is a statement about complexity classes and their relation. Nevertheless, Hirahara's argument seems to employ Meta-Complexity in an essential way, through Kolmogorov complexity and related compression problems. 

Recall that meta-complexity studies the complexity of computational tasks whose inputs are themselves computational objects, such as strings,  or truth tables. A central example is time-bounded Kolmogorov complexity, which measures the length of the shortest program that outputs a given string within a prescribed time bound. The associated meta-computational problem asks for the time-bounded Kolmogorov complexity of a given input string. For more information about meta-complexity, we refer to \citep{SimonsMetaComplexity2023}. For a modern treatment of time-bounded Kolmogorov complexity and its extensions to randomized computation, see \citep{DBLP:journals/eatcs/LuO22}.

Subsequent works have provided simplifications to the original proof of \citep{DBLP:conf/stoc/Hirahara21} and  extensions of the result \citep{GK22, GKLO22, DBLP:conf/itcs/ChenHV22}.  All of these developments continue to rely on meta-complexity. We refer to \citep{DBLP:journals/eatcs/Hirahara22} for a comprehensive treatment of the meta-complexity approach to average-case complexity.

Motivated by these developments and the  intriguing role of meta-complexity, we consider the following natural question:
\medskip
\begin{center}
    \emph{Is Meta-Complexity formally necessary for understanding the relation between\\ average-case and worst-case complexity?}
\end{center}
\medskip
As our main contribution, we give a positive answer to this question. We show that excluding $\PH$ Heuristica (\Cref{eq:excl_PH_Heuristica}) is equivalent to proving a  \emph{gap reduction} between two worst-case meta-computational tasks. This gives a precise sense in which meta-complexity is inevitable: any proof excluding $\PH$ Heuristica can be converted into a proof of the gap reduction from the standard gap meta-complexity problem to the mild-gap meta-complexity problem; conversely, any proof of this gap reduction excludes $\PH$ Heuristica. 

\subsection{Results}

We now state our results in more detail.\footnote{A preliminary version of a subset of these results appeared as part of the unpublished ECCC preprint~\citep{ECCC2025TR215}.}

\paragraph{Average-Case vs Worst-Case Complexity as Meta-Complexity.}
We use $\pK^{t,\mathcal O}(x)$ to denote probabilistic $t$-time bounded Kolmogorov complexity with oracle access to $\mathcal O$. Informally, $\GapMINpKT^{\mathcal O}$ asks one to distinguish strings with small $\pK^{t,\mathcal O}$ complexity from strings whose $\pK^{\poly(t),\mathcal O}$ complexity is larger by a logarithmic term. This is the standard gap setting for several time-bounded Kolmogorov complexity estimation problems in meta-complexity. In contrast, the mild version $\MildGapMINpKT^{\mathcal O}$ asks for the analogous distinction when the larger time bound is only $t+\poly(n)$ rather than $\poly(t)$. We refer to \Cref{sec:preliminaries} for the formal notions. \Cref{fig:gap-instance-inclusions} informally describes the input instances of the two gap problems and their inclusion relations.

\medskip

\begin{figure}[ht]
\centering

\resizebox{0.98\linewidth}{!}{%
\begin{tikzpicture}[
    font=\footnotesize,
    every node/.style={
        align=center,
        inner sep=1pt
    },
    universe/.style={
        fill=gray!8,
        draw=black,
        line width=0.8pt
    },
    yesregion/.style={
        fill=green!18,
        draw=green!55!black,
        line width=0.8pt
    },
    mildregion/.style={
        fill=orange!18,
        draw=orange!65!black,
        line width=0.8pt,
        rounded corners=16pt
    },
    standardregion/.style={
        fill=red!22,
        draw=red!65!black,
        line width=0.8pt,
        rounded corners=12pt
    }
]

    \filldraw[universe]
        (0,2) rectangle (15.4,7.5);

    \node[
        anchor=north west,
        font=\small\bfseries
    ] at (0.25,7.25)
        {$\bm{U=\{0,1\}^n}$};

    \filldraw[yesregion]
        (2.6,4.7) ellipse
        [x radius=1.5, y radius=1.4];

    \node[
        text width=3.5cm
    ] at (2.6,4.7)
        {
        {\small\textbf{YES}}\\[4pt]
        for both problems\\[5pt]
        {\scriptsize $\bm{\pK^t}$ is at most $\bm{s}$}
        };

    \node[
        text width=3.2cm
    ] at (6.25,4.7)
        {
        {\small\textbf{Nonpromise}}\\[4pt]
        for both problems\\[5pt]
        {\scriptsize outside YES and Mild NO}
        };

    \filldraw[mildregion]
        (8.2,2.55) rectangle (14.85,7);

    \node[
        text width=5.8cm,
        font=\small\bfseries
    ] at (11.52,6.62)
        {NO at time $\bm{t+\poly(n)}$};

    \filldraw[standardregion]
        (11.05,3.4) rectangle (14.3,5.9);

    \node[
        text width=2.7cm
    ] at (12.67,4.65)
        {
        {\small\textbf{NO at time $\bm{\poly(t)}$}}\\[5pt]
        NO for both\\[3pt]
        {\scriptsize Standard gap NO}
        };

    \node[
        text width=2.1cm
    ] at (9.55,4.65)
        {
        {\small\textbf{Mild NO}}\\[5pt]
        Standard gap\\
        nonpromise
        };

    \node[
        font=\small\bfseries
    ] at (7.7,1.45)
        {
        $\mathsf{NO}_{\poly(t)}
        \subseteq
        \mathsf{NO}_{t+\poly(n)}$
        };

    \node[
        font=\footnotesize
    ] at (7.7,0.95)
        {
        Standard gap promise
        $=
        \mathsf{YES}
        \cup
        \mathsf{NO}_{\poly(t)}$
        };

    \node[
        font=\footnotesize
    ] at (7.7,0.45)
        {
        Mild gap promise
        $=
        \mathsf{YES}
        \cup
        \mathsf{NO}_{t+\poly(n)}$
        };

\end{tikzpicture}%
}

\caption{
Instance sets for the standard and mild gap problems for fixed
description length bound $s$ and time bound $t$.
Under $\poly(t)\geq t+\poly(n)$, monotonicity of time-bounded
probabilistic Kolmogorov complexity gives
$
\pK^{t+\poly(n),\mathcal O}(x)
\geq
\pK^{\poly(t),\mathcal O}(x).
$
Consequently,
$
\mathsf{NO}_{\poly(t)}
\subseteq
\mathsf{NO}_{t+\poly(n)}.
$
}
\label{fig:gap-instance-inclusions}

\end{figure}
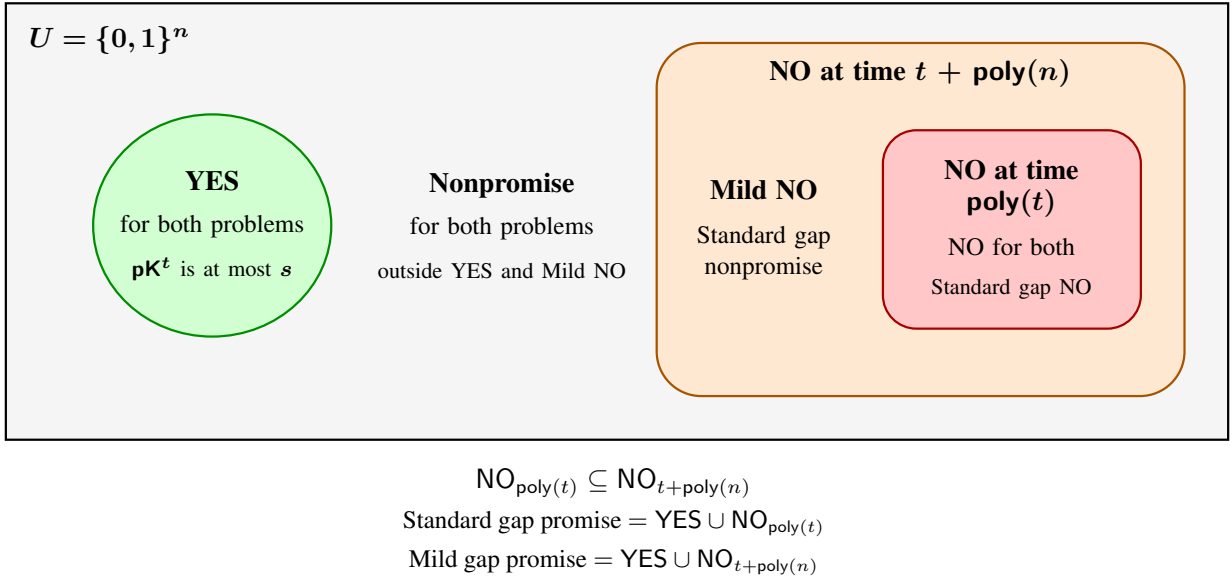

For an oracle-dependent problem $A$ such as $\GapMINpKT$ or $\MildGapMINpKT$, the notation $A^\PH \in \prBPP$ means that the corresponding promise problem is in randomized polynomial time for every oracle language $\mathcal{O} \in \PH$.

\begin{restatable}[Main Equivalence]{theorem}{PHHEUR}\label{t:PH-heuristica-characterization}
	The following statements are equivalent\emph{:}
	\begin{enumerate}
		\item \emph{(Exclusion of PH Heuristica).} $\Dist\mathsf{PH}\subseteq\Avg\BPP\implies \mathsf{PH}\subseteq\BPP$.
		\item  \emph{(Gap Reduction for $\mathsf{pK}^{\mathsf{PH}}$).} $\mathsf{Gap}\textnormal{-}\mathsf{MINpKT}^{\mathsf{PH}} \in \prBPP \implies \mathsf{Mild}\textnormal{-} \mathsf{Gap}\textnormal{-}\mathsf{MINpKT}^{\mathsf{PH}} \in \prBPP$.
	\end{enumerate}
\end{restatable}

Note that the second item refers to the \emph{worst-case} complexities of the corresponding meta-computational problems.
Thus, the average-case versus worst-case complexity of $\PH$ is equivalent to a comparison between two approximation regimes for probabilistic time-bounded Kolmogorov complexity with $\PH$ oracle access. Roughly speaking, the standard gap problem captures average-case complexity, while the mild gap problem captures worst-case complexity. Establishing a gap reduction, namely that an efficient algorithm for the former yields an efficient algorithm for the latter, is exactly what is needed to rule out $\PH$ Heuristica.\footnote{For completeness, we spell out here a reduction that is sufficient for excluding $\PH$ Heuristica. Fix an arbitrary oracle $\mathcal{O} \in \PH$. Exhibit a polynomial $p$ and an oracle $\mathcal{O}' \in \PH$ such that, for every large enough polynomial $q$, separating $\pK^{t, \mathcal{O}}(x) \leq s$ \;vs\; $\pK^{t + p(|x|), \mathcal{O}}(x) \geq s + \log p(|x|)$ reduces in randomized polynomial time to separating $\pK^{t,\mathcal{O'}}(x) \leq s$ \;vs\; $\pK^{q(t), \mathcal{O'}}(x) \geq s + \log q(|x|)$.} The proof of \Cref{t:PH-heuristica-characterization} appears in \Cref{sec:proof_equiv}.

There is a deterministic analogue of the characterization: replacing $\pK^{t,\PH}$ by deterministic $\K^{t,\PH}$ and replacing randomized algorithms by deterministic algorithms gives the corresponding equivalence between
\[
    \Dist\PH\subseteq\Avg\P \Longrightarrow \PH \subseteq \P
\]
and the implication from the deterministic standard gap version of $\MINKT^\PH$ to its deterministic mild gap version. We refer to \Cref{thm:deterministic} for the details.

\paragraph{Eliminating the Gap in the Average-Case Setting.} In our next result, we show that the implication in \Cref{t:PH-heuristica-characterization} (Item 2) can be established in the average-case setting. Indeed, we are able to prove in this case that the easiness of $\mathsf{Gap}\textnormal{-}\mathsf{MINpKT}^{\mathsf{PH}}$ implies the easiness of  $\mathsf{MINpKT}^{\mathsf{PH}}$, where there is \emph{no gap} in the running time when comparing positive and negative instances. We consider both the zero-error setting of $\mathsf{AvgBPP}$ and the two-sided error setting of $\mathsf{HeurBPP}$.
\begin{restatable}[Average-Case Gap Reductions]{theorem}{AVG}\label{t:gap-mild-gap-avg}
The following two equivalences hold:
\begin{itemize}
    \item $
\textnormal{``}\mathsf{Gap}\textnormal{-}\mathsf{MINpKT}^{\mathsf{PH}}\in \mathsf{AvgBPP}\textnormal{''}
\iff
\textnormal{``}\mathsf{MINpKT}^{\mathsf{PH}}\in \mathsf{AvgBPP}\textnormal{''}$.
\item $\textnormal{``}\mathsf{Gap}\textnormal{-}\mathsf{MINpKT}^{\mathsf{PH}}\in \mathsf{HeurBPP}\textnormal{''}
\iff
\textnormal{``}\mathsf{MINpKT}^{\mathsf{PH}}\in \mathsf{HeurBPP}\textnormal{''}$.
\end{itemize}
\end{restatable}
In other words, the time gap between instances can be eliminated in these average-case settings.  
The necessary definitions are presented in \Cref{ss:ave_case_def}, while the proof of \Cref{t:gap-mild-gap-avg} appears in \Cref{sec:average}.

\paragraph{Limits of Relativizing Techniques.} The natural next step is to obtain a gap reduction in the worst case, which by \Cref{t:PH-heuristica-characterization} is both necessary and sufficient to exclude PH Heuristica. We finish with a proof that the implication between the two meta computational tasks in \Cref{t:PH-heuristica-characterization} (Item 2)  cannot be established by relativizing techniques alone.

\begin{theorem}[Relativization Barrier]\label{thm:intro-relativization-barrier}
There is an oracle $\mathcal{O}$ such that
\[
\mathsf{Gap}\textnormal{-}\mathsf{MINpKT}^{\mathsf{PH}^{\mathcal{O}},\mathcal{O}} \in \prBPP^{\mathcal{O}}
\quad\text{but}\quad
\mathsf{Mild}\textnormal{-}\mathsf{Gap}\textnormal{-}\mathsf{MINpKT}^{\mathsf{PH}^{\mathcal{O}},\mathcal{O}} \notin \prBPP^{\mathcal{O}}.
\]
\end{theorem}

Thus, in a relativized world, the standard gap problem can be easy even while the mild gap problem remains hard. We refer to \Cref{ss:ave-case-rel-barrier} for the proof of this result and related discussions.

\subsection{Techniques}\label{ss:intro_techniques}

In this section, we provide an overview of our proofs and techniques.

\paragraph{Main Equivalence.} The proof of \Cref{t:PH-heuristica-characterization} builds heavily on the approaches and techniques from \citep{DBLP:conf/focs/Hirahara20} and \citep{DBLP:journals/eccc/HuL025}. The result follows from two correspondences between \emph{structural complexity} and \emph{meta-complexity}. As explained below, the average-case easiness of $\PH$ is captured by the complexity of the standard gap problem, where a negative instance is allowed to use a much larger time bound $\poly(t)$. In contrast, the worst-case easiness of $\PH$ is captured by the complexity of the mild gap problem, where a negative instance only receives an additive increase of the form $t+\poly(n)$.  Figure~\ref{fig:two_correspondences} provides an overview.

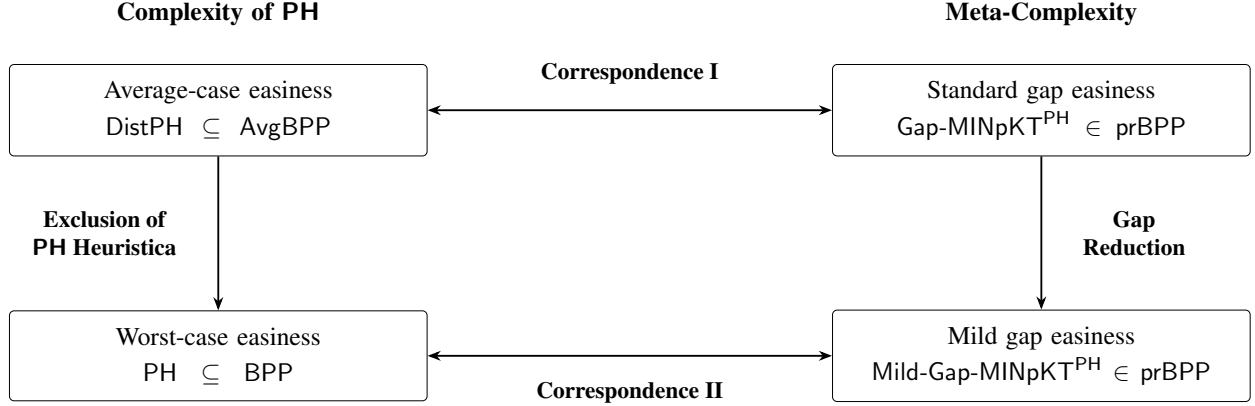
\begin{figure}[ht]
\centering
\resizebox{\textwidth}{!}{%
\begin{tikzpicture}[
    head/.style={font=\bfseries, align=center, text width=5.8cm},
    box/.style={draw, rounded corners=2pt, align=center, text width=5.8cm, minimum height=1.15cm, inner sep=7pt},
    lab/.style={font=\small, align=center, fill=white, inner sep=2pt},
    arr/.style={-{Stealth[length=2mm]}, thick},
    eqarr/.style={{Stealth[length=2mm]}-{Stealth[length=2mm]}, thick}
]
    \node[head] (lefthead) at (-6.2,0) {Complexity of $\bm{\PH}$};
    \node[head] (righthead) at (6.2,0) {Meta-Complexity};

    \node[box] (avg) at (-6.2,-1.45) {Average-case easiness\\[2pt]$\Dist\PH \subseteq \Avg\BPP$};
    \node[box] (worst) at (-6.2,-5.15) {Worst-case easiness\\[2pt]$\PH \subseteq \BPP$};

    \node[box] (gap) at (6.2,-1.45) {Standard gap easiness\\[2pt]$\GapMINpKT^\PH \in \prBPP$};
    \node[box] (mild) at (6.2,-5.15) {Mild gap easiness\\[2pt]$\MildGapMINpKT^\PH \in \prBPP$};

    \draw[eqarr] (avg.east) -- (gap.west);
    \draw[eqarr] (worst.east) -- (mild.west);
    \draw[arr] (avg.south) -- (worst.north);
    \draw[arr] (gap.south) -- (mild.north);

    \node[lab] at ($(avg.east)!0.5!(gap.west)+(0,0.55)$) {\textbf{Correspondence I}};
    \node[lab] at ($(worst.east)!0.5!(mild.west)+(0,-0.55)$) {\textbf{Correspondence II}};
    \node[lab, anchor=east] at ($(avg.south)!0.5!(worst.north)+(-0.55,0)$) {\textbf{Exclusion of}\\\textbf{$\bm{\PH}$ Heuristica}};
    \node[lab, anchor=west] at ($(gap.south)!0.5!(mild.north)+(0.55,0)$) {\textbf{Gap}\\\textbf{Reduction}};

\end{tikzpicture}
}
\caption{Relating the average-case and worst-case complexities of the polynomial hierarchy ($\PH$) is equivalent to a question from meta-complexity: Is there a gap reduction for computing $\mathsf{pK}^{\PH}$?}
\label{fig:two_correspondences}
\end{figure}

The first correspondence identifies average-case easiness of $\PH$ with easiness of the standard gap problem:
\[
    \Dist\PH \subseteq \Avg\BPP
    \quad \Longleftrightarrow \quad
    \GapMINpKT^\PH \in \prBPP.
\]
This correspondence is a randomized extension of Hirahara's characterization of $\Dist\PH\subseteq\Avg\P$ by deterministic worst-case meta-complexity \citep{DBLP:conf/focs/Hirahara20}. We note that in the deterministic setting, a substantial part of the argument from \citep{DBLP:conf/focs/Hirahara20} is the construction of an explicit pseudorandom generator from the assumed easiness of $\GapMINKT^\PH$. This is used not just to conclude $\mathsf{DistPH} \subseteq \Avg\P$ but also as a key technical element of the proof, which manipulates the (deterministic) time-bounded Kolmogorov complexity measure $\K^{t,\PH}$. By working directly with probabilistic Kolmogorov complexity $\pK^{t,\PH}$ and with randomized procedures, we are able to bypasses this specific pseudorandom generator construction.\footnote{As in  \citep{DBLP:conf/focs/Hirahara20}, the remaining steps of the proof still rely on techniques from computational pseudorandomness and on the reconstruction procedure of another pseudorandom generator.} %
Moreover, we provide an induction-based proof that differs somewhat from, and simplifies, the argument of \citep{DBLP:conf/focs/Hirahara20}. We briefly outline the proof below. %

The forward direction follows from a worst-case-to-average-case reduction for estimating time-bounded Kolmogorov complexity. If every language in $\PH$ is easy on every samplable distribution, then $\MINKT^{\mathcal{O}}$ is easy on average for every $\mathcal{O} \in \PH$. This average-case easiness can then be upgraded to a worst-case algorithm for $\GapMINpKT^{\mathcal{O}}$ using standard techniques from meta-complexity.

We now discuss the backward direction. Fix $k_0\geq 1$ and a language $\mathcal{O}$ that is complete for $\Sigma_{k_0}^{\P}$. Assuming that $\GapMINKT^{\mathcal{O}} \in \prBPP$, our goal is to decide any language in $\Sigma_k^{\P}$, for $k\leq k_0$, on input $x$ in time roughly
\begin{equation}\label{eq:cd}
    2^{\pK^{\poly(n),\mathcal{O}}(x)-\K(x)}.
\end{equation}
Since $\K(x)$ is close to $n$ for a typical $x\in\bool^n$, we have $\pK^{\poly(n),\mathcal{O}}(x)-\K(x)\le O(\log n)$
on average. Consequently, \Cref{eq:cd} yields an average-case polynomial-time algorithm for $\PH$.

To establish \Cref{eq:cd}, standard techniques from meta-complexity imply that the assumption $\GapMINKT^{\mathcal{O}} \in \prBPP$ yields a symmetry-of-information statement of the following form: for every $x,y\in\bool^*$ and every $t$,
\[
    \pK^{\poly(t)}(y\mid x)
    \lesssim
    \pK^{t,\mathcal{O}}(x,y)
    -
    \pK^{\poly(t),\mathcal{O}}(x).
\]
Now let $x$ be an instance of the $\Sigma_k^{\P}$ language that we wish to decide, and, if $x$ is a YES instance, let $y^*$ denote the lexicographically smallest witness corresponding to the first existential quantifier. Taking $t$ to be a sufficiently large polynomial, the symmetry-of-information inequality gives
\begin{align*}
    \pK^{\poly(t)}(y^* \mid x) &\lesssim \pK^{t,\mathcal{O}}(x,y^*) - \pK^{\poly(t),\mathcal{O}}(x)  \eqqcolon s\\
    &\lesssim \pK^{2t,\mathcal{O}}(x) - \K(x).
\end{align*}
The second inequality above uses the fact that $y^*$ can be generated efficiently from $x$ given oracle access to $\mathcal{O}$, since $\mathcal{O}$ is complete for a sufficiently high level of $\PH$.

A natural next step is therefore to enumerate all programs of length at most $s$, run them on input $x$, and consider their outputs as candidate witnesses $y$. The difficulty is that checking whether a candidate $y$ is a valid witness requires solving a $\Pi_{k-1}^{\P}$ problem on the instance $(x,y)$. At first sight, one might try to invoke the inductive hypothesis directly. However, doing so would only give a running time bounded by
\[
    2^{\pK^{\poly(n),\mathcal{O}}(x,y)-\K(x,y)}
    \lesssim
    2^{\pK^{\poly(n),\mathcal{O}}(x,y)-\K(x)},
\]
whereas we need the overall running time to remain within
\[
    2^{\pK^{\poly(n),\mathcal{O}}(x)-\K(x)}.
\]
For an arbitrary candidate $y$, there is no reason for
$\pK^{\poly(n),\mathcal{O}}(x,y)$ to be comparable to
$\pK^{\poly(n),\mathcal{O}}(x)$.

The key observation is that we only need to guarantee successful verification of $y^*$, when such a witness exists. For this particular witness, we have
\[
    \pK^{\poly(n),\mathcal{O}}(x,y^*)
    \lesssim
    \pK^{\poly(n),\mathcal{O}}(x),
\]
because $y^*$ can be recovered from $x$ in polynomial time using oracle access to $\mathcal{O}$. Thus, before invoking the inductive-hypothesis algorithm on a candidate $y$, we can first test whether
\[
    \pK^{\poly(n),\mathcal{O}}(x,y)
    \lesssim
    \pK^{\poly(n),\mathcal{O}}(x).
\]
Since $\GapMINKT^\PH\in\prBPP$, this test can be implemented efficiently (up to the usual approximation slack). Candidates that fail the test are discarded. Crucially, $y^*$ passes the test, while every candidate that passes can be verified using the inductive hypothesis within time
\[
    2^{\pK^{\poly(n),\mathcal{O}}(x)-\K(x)}
\]
up to the suppressed polynomial factors. This gives the desired inductive running-time bound.

\medskip

As alluded to above, the second correspondence identifies the worst-case easiness of $\PH$ with the easiness of the mild gap  problem:
\[
    \PH \subseteq \BPP
    \quad \Longleftrightarrow \quad
    \MildGapMINpKT^\PH \in \prBPP.
\]

The easy direction works as follows: it is possible to show that the mild gap problem has an Arthur-Merlin protocol relative to the oracle, and this allows us to prove that if $\PH \subseteq \BPP$ then $\MildGapMINpKT^\PH \in \prBPP$.

The hard direction relies on the witness compression framework of \citep{DBLP:journals/eccc/HuL025}. Suppose $\MildGapMINpKT^\SAT$ is easy. Since $\PH \subseteq \BPP$ follows from $\NP \subseteq \BPP$, it is enough to establish the latter inclusion. Crucially, we derive from the easiness assumption the following consequence: For every $\NP$ verifier, every yes instance $x$, and every witness $y$ for $x$,
\begin{equation}\label{eq:intro_witness_compression}
    \pK^{\poly(n)}(y \mid x) \le \log |W_x| + O(\log n),
\end{equation}
where $W_x$ is the witness set of $x$. Therefore, by choosing a random short program and random coins, one samples an accepting witness with inverse polynomial probability. Repetition then gives a $\BPP$-algorithm for every $\NP$-language, as desired.

An important ingredient in the proof of  \Cref{eq:intro_witness_compression} in \citep{DBLP:journals/eccc/HuL025}  is that probabilistic nondeterministic Kolmogorov complexity ($\mathsf{pnK}^t$) admits a \emph{language compression theorem}: witnesses for an $\NP$ instance can be described using essentially the logarithm of the number of witnesses. Our first observation is that $\pK^{t,\SAT}$ is strong enough to inherit this compression property, since SAT oracle access can simulate the nondeterministic choices appearing in $\pnK^t$ with only polynomial overhead. Note that this is not sufficient to prove \Cref{eq:intro_witness_compression}, where we need an upper bound on $\pK^t$ instead of $\pK^{t,\SAT}$. 

The remaining task is precisely to remove the $\SAT$ oracle from the description of $y$ given $x$. The key ingredient here is that the assumption $\MildGapMINpKT^\PH \in \prBPP$ provides an ordinary probabilistic polynomial-time algorithm $A$ for testing whether a string has low $\pK^{t,\SAT}$ complexity. Although the complexity measure tested by $A$ allows its programs to query $\SAT$, the algorithm $A$ itself has no oracle access. Consequently, if $A$ can be used as a distinguisher in a reconstruction argument, the resulting reconstruction procedure is also an ordinary probabilistic procedure. It therefore yields a bound on the unrelativized complexity $\pK^{\poly(t)}(y\mid x)$, rather than merely on $\pK^{\poly(t),\SAT}(y\mid x)$. This is the mechanism by which the proof removes the $\SAT$ oracle.

Concretely, we compare the $\SAT$-oracle $\pK$ complexity of $(x,y)$ in two ways. The upper bound describes $x$ using a $\SAT$ oracle program and then describes $y$ given $x$ using another such program:
\[
\pK^{t+\poly(n),\SAT}(x,y)
\leq
\pK^{t,\SAT}(x)
+
\pK^{\poly(n),\SAT}(y\mid x)
+
O(\log t).
\]
For the lower bound, fix $k$ slightly below the ordinary complexity $\pK^{\poly(t)}(y\mid x)$ and compare the distributions $(x,\mathcal U_{mk+k})$ and $(x,\mathsf{DP}_k(y,\mathcal U_{mk}))$, where $\mathsf{DP}$ (Direct Product) is a pseudorandom generator with suitable pseudorandomness and reconstruction properties. Weak symmetry of information implies that a uniformly random extension of $x$ usually has high $\SAT$ relative complexity, so $A$ rejects it at an appropriate threshold. If $A$ behaved noticeably differently on the generated extension, the reconstruction property of the direct product generator would turn $A$ into an ordinary probabilistic description of $y$ given $x$ of length at most $k+O(\log t)$, contradicting the choice of $k$. Thus, $A$ must also reject some generated extension. With some additional work, this can be used to prove the following inequality:
\[
\pK^{t+\poly(n),\SAT}(x,y)
+
O(\log t)
\geq
\pK^{\poly(t)}(y\mid x) + \pK^{\poly(t),\SAT}(x).
\]
The important asymmetry is that the term involving $x$ remains $\SAT$ relative, while the term involving $y \mid x$ is now unrelativized.

Combining the upper and lower bounds then rearranging gives
\[
\pK^{\poly(t)}(y\mid x)
\leq
\pK^{\poly(n),\SAT}(y\mid x)
+
\Big ( \pK^{t,\SAT}(x) - 
\pK^{\poly(t),\SAT}(x) \Big )
+
O(\log t).
\]
Thus, the $\SAT$ oracle has been removed from the description of $y$, at the cost of a \emph{computational depth term} measuring the decrease in the $\SAT$ relative complexity of $x$ between the two time bounds $t$ and $\poly(t)$. 

To eliminate this remaining cost, we use a relativized extension of the lifted measure $\ell\lpK_{\lambda,\gamma}^t$ from \citep{DBLP:journals/eccc/HuL025}. Crucially, this measure admits a more efficient form of symmetry of information, allowing the two relevant time bounds to differ by only an additive $\poly(n)$ amount. The assumed algorithm for $\MildGapMINpKT^\SAT$ lifts to the corresponding meta-computational problem for the relativized $\ell\lpK$ variant considered here, and a standard telescoping argument now finds a polynomially bounded time scale at which the lifted complexity of $x$ changes by only $O(1)$. At this time scale, the two $\SAT$ relative terms involving $x$ cancel up to logarithmic overhead, leaving
\[
\pK^{\poly(n)}(y\mid x)
\leq
\pK^{\poly(n),\SAT}(y\mid x)
+
O(\log n).
\]
Combining this oracle removal inequality with language compression for $\pK^{t,\SAT}$ is sufficient to prove  \Cref{eq:intro_witness_compression}.

\medskip

Finally, the two correspondences described above yield \Cref{t:PH-heuristica-characterization}.

\paragraph{Equivalence of Gap and No-Gap Hardness in the Average-Case Setting.} Here we describe the ideas behind the proof of \Cref{t:gap-mild-gap-avg}. It is not hard to see that the tractability of $\mathsf{MINpKT}^{\mathsf{PH}}$ implies that of $\mathsf{Gap}\textnormal{-}\mathsf{MINpKT}^{\mathsf{PH}}$. Therefore, it suffices to prove the converse implication in each item.

We first consider the zero-error setting of $\mathsf{AvgBPP}$. The starting point is the observation that standard worst-case-to-average-case reduction techniques for meta-complexity problems, developed in~\cite{DBLP:conf/focs/Hirahara18,GKLO22}, allow us to obtain that if $\mathsf{Gap}\textnormal{-}\mathsf{MINpKT}^{\mathsf{PH}}$ admits a zero-error average-case algorithm, then it also admits a worst-case algorithm. In other words, $\mathsf{Gap}\textnormal{-}\mathsf{MINpKT}^{\mathsf{PH}} \in \mathsf{prBPP}$.
Then by Correspondence~I (illustrated in Figure~1), this yields $\Dist\PH \subseteq \Avg\BPP$.
Finally, we observe that the non-gap problem, $\mathsf{MINpKT}^{\mathsf{PH}}$, is decidable in (promise) $\AM^{\mathsf{PH}}=\mathsf{PH}$. Consequently, $\mathsf{MINpKT}^{\mathsf{PH}}$ also admits a zero-error average-case algorithm.

We now turn to the two-sided-error setting of $\mathsf{HeurBPP}$. First, it is worth noting that existing work implicitly establishes that the (two-sided-error) average-case tractability of $\mathsf{Gap}\textnormal{-}\mathsf{MINpKT}$ and $\mathsf{MINpKT}$ are equivalent. Indeed,~\cite{DBLP:conf/stoc/IlangoRS22} showed that the (two-sided-error) average-case hardness of approximating the time-unbounded Kolmogorov complexity $\K$ is equivalent to the existence of one-way functions. Their proof can be adapted to $\pK^t$ to show that the (two-sided-error) average-case hardness of $\mathsf{Gap}\textnormal{-}\mathsf{MINpKT}$ is equivalent to the existence of one-way functions (and hence to the existence of pseudorandom generators~\cite{DBLP:journals/siamcomp/HastadILL99}).\footnote{In fact, here we will only use the implication that if one can efficiently approximate $\pK^t$-complexity, then pseudorandom generators do not exist. The intuition is that such an approximation would allow us to break any candidate pseudorandom generator: the output of a generator has small time-bounded Kolmogorov complexity, roughly bounded by the length of its seed, whereas a truly random string has high complexity.
For the converse direction, the result in \cite{DBLP:conf/stoc/IlangoRS22} relies on the coding theorem for $\K$-complexity. The same argument extends to $\pK^t$ by applying the optimal coding theorem for $\pK^t$-complexity established in \cite{DBLP:conf/icalp/LuOZ22}.} On the other hand,~\cite{DBLP:conf/crypto/LiuP23} established that the (two-sided-error) average-case hardness of $\mathsf{MINpKT}$ is also equivalent to the existence of one-way functions.

It is therefore natural to try to extend these equivalences to the setting with $\mathsf{PH}$ oracles by considering $\PH$-computable one-way functions. However, there are a couple of obstacles. First, it is unclear how to adapt the ideas of~\cite{DBLP:conf/stoc/IlangoRS22} to show that if $\mathsf{Gap}\textnormal{-}\mathsf{MINpKT}^{\mathsf{PH}}$ is easy on average, then no $\mathsf{PH}$-computable pseudorandom generators (equivalently, $\mathsf{PH}$-computable one-way functions) exist. Ideally, one would use an average-case algorithm for $\mathsf{Gap}\textnormal{-}\mathsf{MINpKT}^{\mathsf{PH}}$ to break the security of such a generator, since every output of a $\mathsf{PH}$-computable generator has small time-bounded Kolmogorov complexity with $\mathsf{PH}$ oracles. The issue is that a $\mathsf{PH}$-computable pseudorandom generator induces a distribution that is polynomial-time samplable only with access to a $\mathsf{PH}$ oracle, whereas our assumption asserts only that $\mathsf{Gap}\textnormal{-}\mathsf{MINpKT}^{\mathsf{PH}}$ is easy over polynomial-time samplable distributions without such oracle access. Second, the techniques of~\cite{DBLP:conf/crypto/LiuP23}  yield only the following relativized statement: for every oracle $\mathcal{O}$ (for example, $\mathcal{O}\in\mathsf{PH}$), if there are no one-way functions computable with oracle access to $\mathcal{O}$ that are secure against probabilistic polynomial-time algorithms \emph{with oracle access} to $\mathcal{O}$, then $\mathsf{MINpKT}^{\mathcal{O}} \in \mathsf{HeurBPP}^{\mathcal{O}}$.
Thus, these results do not yield the desired equivalence between the average-case tractability of $\mathsf{Gap}\textnormal{-}\mathsf{MINpKT}^{\mathsf{PH}}$ and that of $\mathsf{MINpKT}^{\mathsf{PH}}$ via $\mathsf{PH}$-computable one-way functions.

Instead, we take a different route based on~\cite{DBLP:conf/innovations/HiraharaS22a}. Specifically,~\cite{DBLP:conf/innovations/HiraharaS22a} proved that
\[
\Dist\PH \not\subseteq \mathsf{HeurBPP}
\]
is equivalent to the existence of $\mathsf{PH}$-computable pseudorandom generators that are secure against probabilistic polynomial-time algorithms. Moreover, the pseudorandom generators obtained there satisfy a certain entropy-preserving property that allows us to break them using an average-case algorithm for $\mathsf{Gap}\textnormal{-}\mathsf{MINpKT}^{\mathsf{PH}}$ over the uniform distribution. Consequently, the (two-sided-error) average-case tractability of $\mathsf{Gap}\textnormal{-}\mathsf{MINpKT}^{\mathsf{PH}}$ implies that
\[
\Dist\PH \subseteq \mathsf{HeurBPP}.
\]
Again, since the non-gap problem $\mathsf{MINpKT}^{\mathsf{PH}}$ is decidable in (promise) $\AM^{\mathsf{PH}}=\mathsf{PH}$, it follows that $\mathsf{MINpKT}^{\mathsf{PH}}$ also admits a two-sided-error average-case algorithm.

One technical point omitted from the discussion above is that the characterization of~\cite{DBLP:conf/innovations/HiraharaS22a} only concerns pseudorandom generators with arbitrary logarithmic stretch. A minor extension of their result is nevertheless sufficient for our purposes, because it is enough to establish the desired connection with the tractability of $\mathsf{Gap}_{\rho}\textnormal{-}\mathsf{MINpKT}^{\mathsf{PH}}$, where $\rho$ is any fixed polynomial (see \Cref{t:prg-pK} for details).

\paragraph{Relativization Barrier.} Finally, we briefly comment on the proof of \Cref{thm:intro-relativization-barrier}. 
The relativization barrier is proved by applying the same correspondences described above inside a suitable oracle world. Impagliazzo \citep{DBLP:conf/coco/Impagliazzo11}, followed by Hirahara and Nanashima \citep{DBLP:conf/focs/HiraharaN21}, give an oracle $\mathcal{O}$ for which $\Dist\PH^\mathcal{O}\subseteq\Avg\BPP^\mathcal{O}$ but $\PH^\mathcal{O}\nsubseteq\BPP^\mathcal{O}$. The first property feeds into the \emph{relativized} average-case correspondence, giving $\mathsf{Gap}\textnormal{-}\mathsf{MINpKT}^{\mathsf{PH}^{\mathcal{O}},\mathcal{O}} \in \prBPP^{\mathcal{O}}$. On the other hand, if the corresponding mild gap problem were also easy, the \emph{relativized} worst-case correspondence would imply $\NP^\mathcal{O}\subseteq\BPP^\mathcal{O}$ and hence $\PH^\mathcal{O}\subseteq\BPP^\mathcal{O}$, contradicting the oracle separation. Most of the work goes into the verification that the techniques employed to establish these correspondences relativize. 

\subsection{Related Work}\label{ss:related_work}

The modern theory of average-case complexity for distributional problems was initiated by Levin \citep{DBLP:journals/siamcomp/Levin86} and developed into the framework surveyed by Bogdanov and Trevisan \citep{DBLP:journals/fttcs/BogdanovT06}. Impagliazzo's five worlds \citep{DBLP:conf/coco/Impagliazzo95} placed the average-case versus worst-case question at the center of complexity theory and cryptography. The particular question relevant here is the exclusion of $\PH$ Heuristica, namely whether $\Dist\PH\subseteq\Avg\BPP$ implies $\PH\subseteq\BPP$.

Classical barrier results explain why this question has resisted the most direct reduction based approaches. Feigenbaum and Fortnow \citep{DBLP:journals/siamcomp/FeigenbaumF93} showed that strong forms of random self reducibility for complete sets would imply unexpected collapses of the polynomial hierarchy. Bogdanov and Trevisan \citep{DBLP:conf/focs/BogdanovT03} extended this barrier to nonadaptive self correction with respect to samplable distributions: if an $\NP$ complete problem had such a self corrector, then $\coNP\subseteq\NP/\poly$ and the polynomial hierarchy collapses. Thus, black box nonadaptive reductions cannot by themselves base average-case hardness in $\NP$ on worst case hardness of an $\NP$-complete problem, except under unlikely collapses. Gutfreund, Shaltiel, and Ta Shma \citep{DBLP:journals/cc/GutfreundST07} obtained a positive but weaker connection: if $\NP\nsubseteq\BPP$, then for every fixed algorithm there is a polynomial time samplable distribution on which it errors often. The caveat is that the hard distribution may depend on the algorithm, so this does not produce a single distributional problem hard for all efficient algorithms.

Hirahara developed the meta-complexity approach to this question through a series of works (see, e.g.,~\citep{DBLP:journals/eatcs/Hirahara22}). The first step was a non-black-box worst-case to average-case reduction for time-bounded Kolmogorov complexity \citep{DBLP:conf/focs/Hirahara18}. Hirahara then characterized $\Dist\PH\subseteq\Avg\P$ by the worst-case easiness of $\mathsf{Gap}\textnormal{-}\mathsf{MINKT}^{\PH}$ \citep{DBLP:conf/focs/Hirahara20}, and later proved that average-case easiness of $\PH$ yields worst case algorithms for $\PH$ running in time $2^{O(n/\log n)}$ \citep{DBLP:conf/stoc/Hirahara21}. Goldberg, Kabanets, Lu, and Oliveira \citep{GKLO22} lifted the original result to the probabilistic setting and simplified the proof using probabilistic Kolmogorov complexity.

Other works have refined and extended this paradigm. Chen, Hirahara, and Vafa \citep{DBLP:conf/itcs/ChenHV22} obtained average-case hardness consequences for $\NP$ and $\PH$ from fine-grained worst-case hardness assumptions. Hirahara \citep{DBLP:conf/coco/Hirahara22} related the exclusion of Heuristica to symmetry of information and hardness of approximating conditional time-bounded Kolmogorov complexity. 

Our work is different in that the main theorem offers an \emph{equivalence} for the  implication $\Dist\PH\subseteq\Avg\BPP\Rightarrow\PH\subseteq\BPP$, rather than only a sufficient condition. Moreover, the equivalence is between this complexity-theoretic implication and a reduction between two meta-computational approximation regimes, the standard gap and mild gap versions of $\MINpKT^\PH$.

Relativization has also played an important role in understanding what meta-complexity methods can and cannot prove. Impagliazzo \citep{DBLP:conf/coco/Impagliazzo11} constructed relativized separations between worst-case and average-case complexity for $\NP$, and Hirahara and Nanashima \citep{DBLP:conf/focs/HiraharaN21} studied worst-case learning in relativized Heuristica. Ren and Santhanam \citep{DBLP:conf/stacs/RenS22} initiated a systematic study of relativization in meta-complexity. Among other oracle worlds, they separate decision and search versions of $\mathsf{MCSP}$, show that different $\mathsf{MCSP}$ thresholds can lead to different complexities, that $\Kt$ can become efficiently approximable, and that robust reductions from $\Dist\NP$ to $\mathsf{Gap}\textnormal{-}\mathsf{MINKT}$ face oracle barriers. Hirahara, Lu, and Ren \citep{DBLP:conf/coco/HiraharaLR23} later introduced bounded relativization as a finer way to classify proof techniques that may go beyond ordinary relativization while still relativizing against restricted oracle classes.

In contrast, our relativization theorem deals with a delicate regime of parameters and constructs an oracle world where the standard gap version of $\MINpKT^\PH$ is easy while the mild gap version remains hard. In this sense, the oracle result isolates the exact nonrelativizing step needed by our main equivalence: transforming average-case easiness into worst-case easiness of the standard gap problem is relativizing, but upgrading the standard gap regime to the mild gap regime cannot be done by relativizing arguments alone.

Nonrelativizing techniques have also played an important role in establishing hardness results for meta-computational problems. Hirahara \citep{DBLP:conf/focs/Hirahara22} employed ideas from algorithmic information theory to overcome a relativization barrier for showing hardness of learning programs and, in the same work, proved $\NP$-hardness under randomized polynomial-time reductions for the partial function variant $\mathsf{MCSP}^*$ and related problems such as $\mathsf{MKTP}^*$ and $\mathsf{MINKT}^*$. The proof uses conditional Kolmogorov complexity and pseudorandom generators to quantify the information available to a small program, providing a complementary example of nonrelativizing techniques in meta-complexity.

A separate line of work has sought to reduce the \emph{compression gap} between the description length thresholds in the yes and no cases of meta-computational promise problems. Beginning with Hirahara's non-black-box worst-case to average-case reduction for time-bounded Kolmogorov complexity, subsequent results progressively reduced this loss \citep{DBLP:conf/focs/Hirahara18,DBLP:conf/stoc/Hirahara21,GKLO22}. Most recently, Liu and Pass \citep{DBLP:conf/innovations/LiuP26} eliminated the compression gap entirely in the large threshold regime, where the description length threshold is close to the input length, obtaining a worst 
-case to average-case reduction for the exact problem under a standard derandomization assumption. This progress concerns the gap between description lengths. By contrast, reducing the separate \emph{time gap} between the running time bounds in the two promise cases appears to have received less attention. Investigation of changes in the time bound appear in the fine-grained reductions of Chen, Hirahara, and Vafa \citep{DBLP:conf/itcs/ChenHV22} and in the recent witness compression framework of Hu, Lu, and Oliveira \citep{DBLP:journals/eccc/HuL025}. Our work shows that the \emph{time gap} is fundamental to our understanding of average-case complexity.

\vspace{0.6cm}

\noindent \textbf{Use of AI Tools.} All ideas, arguments, and proofs presented in this paper were developed by the authors. The proofs were checked for correctness using GPT 5.6 and GPT 6.0.\\

\noindent \textbf{Acknowledgements.}
 This work received support from the UKRI Frontier Research Guarantee Grant EP/Y007999/1 and the Centre for Discrete Mathematics and its Applications (DIMAP) at the University of Warwick. The authors would also like to thank the Isaac Newton Institute for Mathematical Sciences, Cambridge, for support and hospitality during the programme Logical Foundations of Complexity Theory (LFC), where part of the work on this paper was undertaken. The LFC programme was supported by EPSRC grant EP/Z000580/1.

\section{Preliminaries}\label{sec:preliminaries}

We use $\epsilon$ to denote the empty string. The length of a string $x$ is denoted $|x|$. Given strings $x \in \{0,1\}^n$ and $y \in \{0,1\}^m$, we let $x \circ y \in \{0,1\}^{n + m}$ denote their concatenation. 

We write $x \sim \mathcal{D}$ to denote that $x$ is sampled  according to the distribution $\mathcal{D}$. For an element $a$ in the support of $\mathcal{D}$, we let $\mathcal{D}(a)$ denote its probability. We use $\mathcal{U}_\ell$ to denote the uniform distribution over $\{0,1\}^\ell$.

\subsection{Computational Models}

We consider multi-tape Turing machines over the alphabet $\Sigma = \{0,1,\bot\}$. In addition to the usual input tape and work tapes, in some contexts we allow a machine to access three additional tapes: one holding an auxiliary string, one holding a non-deterministic guess, and one holding the randomness. We assume that the machine has \emph{sequential access} (i.e., no random access) to the randomness tape and the non-deterministic tape.\footnote{To avoid confusion, sequential access means that the tape head can remain at a given tape cell or move left or right at each step of the computation. In particular, we do not require read-only or one-way access to the tape.}

For a machine $M$, time bound $t$, and string $x$, we let $\mathsf{K}^t_M(x)$ be the minimum length of a string $p \in \{0,1\}^*$ such that $M$ outputs $x$ on input string $p$ when running for at most $t$ steps. We say that a (universal) machine $U$ is \emph{time-efficient} if for every machine $U'$ there is a constant $C$ such that, for every string $x$ and time bound $t$, we have $\mathsf{K}_U^{C \cdot t \log t}(x) \leq \mathsf{K}_{U'}^{t}(x) + C$. It is known that time-efficient machines exist (see, e.g.,~\citep{DBLP:series/txcs/LiV19}). We fix such a machine $U$ when specifying Kolmogorov complexity.

We also assume an encoding function $\langle \cdot \rangle$ for machines (which we omit for notational simplicity) such that $U$, when given (the encoding of) a machine $M$, an input $\alpha$, an auxiliary string $x$, a non-deterministic guess $w$, and a random string $r$, runs in time $\widetilde{O}(|M|, t)$ and outputs the string produced by $M(\alpha; x; w; r)$, where $|M|$ denotes the encoding length of $M$ and $t$ is the number of steps $M$ runs on $(\alpha; x; w; r)$.

A \emph{program} $\Pi$ is defined as a pair $\Pi \vcentcolon= (M, \alpha)$, where $M$ is a valid encoding of a machine and $\alpha \in \bool^*$ is an ``advice'' string. The output of $\Pi$ on $(x; w; r)$ (which corresponds to the auxiliary input, the non-deterministic guess, and the randomness), denoted $\Pi(x; w; r)$, is given by $M(\alpha; x; w; r)$.

The \emph{running time} of the program $\Pi$ on input $(x, w, r)$ is defined to be the number of steps $M$ runs on input $\mleft(\alpha; x; w; r\mright)$.

We also define the \emph{size} of the program $\Pi$. To encode a machine-advice pair $(M, \alpha)$, we use the following format:
\[
s_1 \,01\, \langle M \rangle\, s_2 \,01\, \alpha,
\]
where $\langle M \rangle$ is an encoding of the machine $M$,\footnote{As noted above, we assume a fixed standard encoding for Turing machines.} and $s_1$ (resp.\ $s_2$) is the binary representation of the integer specifying the length of $\langle M \rangle$ (resp.\ $\alpha$), with each bit duplicated. The size of $\Pi$, denoted $|\Pi|$, is the total bit length of this description.

We say that a program $\Pi$ \emph{non-deterministically} outputs $y$ on input $(x; r)$ in time $t$ if:
\begin{itemize}
    \item $\forall w \in \{0,1\}^t$, $\Pi(x; w; r)$ outputs either $y$ or $\bot$ in time $t$, and
    \item $\exists w \in \{0,1\}^t$, $\Pi(x; w; r)$ outputs $y$ in time $t$.
\end{itemize}

\paragraph{A Computational Model with Efficient Composition.}

The above computational model suffices for most of our results, except for \Cref{t:PH-heuristica-characterization}. As in the main result of \cite{DBLP:journals/eccc/HuL025}, \Cref{t:PH-heuristica-characterization} requires a computational model that possesses an additional natural efficiency property called ``efficient composition''. Specifically, this property states that there exists a polynomial $p$ such that the following holds. For any $x, y, z, r \in \{0,1\}^*$, if:
\begin{itemize}
    \item there is a program of size $s_1$ that  outputs $y$ on input $(x; r)$ in time $t_1$, and
    \item there is a program of size $s_2$ that  outputs $z$ on input $y$ in time $t_2$,
\end{itemize}
then there is a program of size $s_1 + s_2 + \log p(|x| + |y| + |z|)$ that  outputs $z$ on input $(x; r)$ in time $t_1 + p(t_2) + p(|x| + |y| + |z|)$. (As discussed in detail in \cite{DBLP:journals/eccc/HuL025}, this property is typically implicitly assumed in the literature on time-bounded Kolmogorov complexity.)

To guarantee the efficient composition property, we consider \emph{a sequence of machine-advice pairs}. More formally, a program $\Pi$ is now defined as
\[
\Pi \vcentcolon= \left( (M_1, \alpha_1), \dots, (M_\ell, \alpha_\ell) \right),
\]
where each $M_i$ is a valid encoding of a machine and $\alpha_i \in \bool^*$. The output of $\Pi$ on $(x; r)$ (which corresponds to the auxiliary input and the randomness), denoted $\Pi(x; r)$, is obtained as follows:\footnote{For a string $z$ of length $\ell$ and a subset $S \subseteq [\ell]$, where $[\ell] = \{1, 2, \ldots, \ell\}$, we let $z_S$ denote the substring of $z$ obtained by concatenating the bits of $z$ at indices in $S$. For $a < b$, we let $[a:b]$ denote the set $\{a, a+1, \ldots, b\}$.}
\begin{itemize}
    \item Let $y_1$ be the output of $M_1(\alpha_1; x; r)$. Let $t_1$ be the number of steps $M_1(\alpha_1; x; r)$ takes to produce $y_1$.
    \item Let $y_2$ be the output of $M_2(\alpha_2; x \circ y_1; r_{[t_1+1:|r|]})$, and let $t_2$ be the corresponding running time.
    \item[] \begin{center}{$\vdots$}\end{center}
    \item Let $y_\ell$ be the output of $M_\ell\mleft(\alpha_\ell; x \circ y_{\ell-1}; r_{\mleft[1+\sum_{j<\ell}t_{j}:|r|\mright]}\mright)$, and let $t_\ell$ be the corresponding running time.
    \item Output $y_\ell$.
\end{itemize}

The running time of the program $\Pi$ on input $(x; r)$ is given by $\sum_{i=1}^{\ell} t_i$.

We now describe how the program $\Pi$ is encoded. A machine-advice pair is encoded as described before. The full description of $\Pi$ is obtained by concatenating the encodings of all pairs $(M_i, \alpha_i)$ for $i \in [\ell]$. The size of $\Pi$, denoted $|\Pi|$, is the total bit length of this description.

It can be shown that this definition of a program satisfies the efficient composition property (see \cite[Proposition 57]{DBLP:journals/eccc/HuL025}). The definition can be extended to programs that have oracle access to a language $\mathcal{O}$ in the natural way.

\subsection{Time-Bounded Kolmogorov Complexity}\label{sec:prelim_Kolm_complex}

We record only the Kolmogorov notions used in the main results and their proofs.

\begin{definition}[$\mathsf{K}^t$]
For $x,y,r\in\bool^*$, $t\in\mathbb{N}$, and an oracle $\mathcal{O}$, define
\[
\K^{t,\mathcal{O}}(x\mid y; r) \eqdef \min_{\Pi \in \{0,1\}^*} \mleft\{|\Pi|  \,\,\middle\vert\,\,  \Pi^{\mathcal{O}}(y;\epsilon;r) \textnormal{ outputs $x$ within $t$ steps}\mright\}.
\]
We write $\K^{t,\mathcal{O}}(x\mid y)$ for $\K^{t,\mathcal{O}}(x\mid y;\epsilon)$ and $\K^{t,\mathcal{O}}(x)$ for $\K^{t,\mathcal{O}}(x\mid \epsilon)$. We write $\K(x)$ for the corresponding unbounded plain Kolmogorov complexity. If oracle queries are non adaptive, we write $\K^{t,\parallel\mathcal{O}}$.
\end{definition}

\begin{definition}[$\mathsf{pK}^t$]
For $x,y\in\bool^*$, $\lambda\in [0,1]$, $t\in\mathbb{N}$, and an oracle $\mathcal{O}$, define
\[
\pK^{t,\mathcal{O}}_\lambda(x\mid y)\eqdef\min\mleft\{s\in\mathbb{N}\;\middle\vert\;\Prob_{r\sim\{0,1\}^t}\mleft[\K^{t,\mathcal{O}}(x\mid y;r)\le s\mright]\ge\lambda\mright\}.
\]
We omit $\mathcal{O}$ when it is empty and omit $\lambda$ when $\lambda=2/3$.
\end{definition}

\begin{definition}[$\mathsf{nK}^t$]
For $x,y,r\in\bool^*$ and $t\in\mathbb{N}$, define
\[
\nK^t(x\mid y;r) \eqdef \min_{\Pi \in \{0,1\}^*} \mleft\{|\Pi|  \,\,\middle\vert\,\,  
\begin{aligned}
&\forall w\in\bool^t, \Pi(y;w;r) \textrm{ outputs $x$ or $\bot$ in $t$ steps}; \\
&\exists w \in \{0, 1\}^t, \Pi(y;w;r) \textrm{ outputs $x$ in $t$ steps}
\end{aligned}
\mright\}.
\]
We write $\nK^t(x\mid y)$ for $\nK^t(x\mid y;\epsilon)$.
\end{definition}

\begin{definition}[$\mathsf{pnK}^t$]\label{def:pnK-t}
For $x,y\in\bool^*$, $\lambda\in [0,1]$, and $t\in\mathbb{N}$, define
\[
\pnK^t_\lambda(x\mid y)\eqdef\min\mleft\{s\in\mathbb{N}\;\middle\vert\;\Prob_{r\sim\{0,1\}^t}\mleft[\nK^t(x\mid y;r)\le s\mright]\ge\lambda\mright\}.
\]
We omit $\lambda$ when $\lambda=2/3$.
\end{definition}

\subsection{Meta-Complexity}

We now define the promise problems that appear in the theorem statements. The deterministic problems for the deterministic analogue are included  in \Cref{t:PH-heuristica-characterization-det}.

\begin{definition}[$\mathsf{MINKT}$]
    For an oracle $\mathcal{O}$, we define the promise language $\mathsf{MINKT}^{\mathcal{O}} \eqdef (\mathsf{YES}, \mathsf{NO})$ as
    \[
        \left\{
        \begin{aligned}
            \mathsf{YES} & = \mleft\{(x,1^s,1^t) \mid \K^{t,\mathcal{O}}(x)\leq s\mright\} \\
            \mathsf{NO} & = \mleft\{(x,1^s,1^t) \mid \K^{t,\mathcal{O}}(x)> s\mright\}
        \end{aligned}
        \right..
    \]
\end{definition}

\begin{definition}[$\mathsf{Gap}\textnormal{-}\mathsf{MINKT}$]\label{def:gap-MINKT}
        For an oracle $\mathcal{O}$ and $\rho\colon\mathbb{N} \to \mathbb{N}$, we define the promise language
        \[
\mathsf{Gap}_{\rho}\textnormal{-}\mathsf{MINKT}^{\mathcal{O}}\eqdef (\mathsf{YES}, \mathsf{NO})
        \]
        as
    \[
        \left\{
        \begin{aligned}
            \mathsf{YES} & = \mleft\{(x,1^s,1^t) \mid \K^{t, \mathcal{O}}(x)\leq s\mright\} \\
            \mathsf{NO} & = \mleft\{(x,1^s,1^t) \mid \K^{\rho(t), \mathcal{O}}(x)\geq s + \log\rho(t)\mright\}
        \end{aligned}
        \right..
    \]
    For a complexity class $\mathcal{C}$, we say that $\mathsf{Gap}\textnormal{-}\mathsf{MINKT}^{\mathcal{O}} \in \mathcal{C}$ if there exists a constant $c > 0$ such that $\mathsf{Gap}_\rho\textnormal{-}\mathsf{MINKT}^{\mathcal{O}} \in \mathcal{C}$ with $\rho(n) = n^c$.

\end{definition}

\begin{definition}[$\mathsf{Mild}\textnormal{-}\mathsf{Gap}\textnormal{-}\mathsf{MINKT}$]\label{def:mild-gap-MINKT}
        For an oracle $\mathcal{O}$ and a function $\rho\colon\mathbb{N}\to \mathbb{N}$, we define the promise language $\mathsf{Mild}\textnormal{-}\mathsf{Gap}_{\rho}\textnormal{-}\mathsf{MINKT}^{\mathcal{O}}\eqdef (\mathsf{YES}, \mathsf{NO})$ as
    \[
        \left\{
        \begin{aligned}
            \mathsf{YES} & = \mleft\{(x,1^s,1^t) \mid \K^{t, \mathcal{O}}(x)\leq s \mright\} \\
            \mathsf{NO} & = \mleft\{(x,1^s,1^t) \mid \K^{t + \rho(|x|), \mathcal{O}}(x)\geq s + \log\rho(t)\mright\}
        \end{aligned}
        \right..
    \]
    For a complexity class $\mathcal{C}$, we say that $\mathsf{Mild}\textnormal{-}\mathsf{Gap}\textnormal{-}\mathsf{MINKT}^{\mathcal{O}} \in \mathcal{C}$ if there exists a constant $c > 0$ such that $\mathsf{Mild}\textnormal{-}\mathsf{Gap}_\rho\textnormal{-}\mathsf{MINKT}^{\mathcal{O}} \in \mathcal{C}$ with $\rho(n) = n^c$.
    
\end{definition}

\begin{definition}[$\MINpKT$]
    For an oracle $\mathcal{O}$, we define the promise language $\MINpKT^{\mathcal{O}} \eqdef (\mathsf{YES}, \mathsf{NO})$ as
    \[
        \left\{
        \begin{aligned}
            \mathsf{YES} & = \mleft\{(x,1^s,1^t,1^a,1^b) \mid \pK^{t,\mathcal{O}}_{b/a}(x)\leq s\mright\} \\
            \mathsf{NO} & = \mleft\{(x,1^s,1^t,1^a,1^b) \mid \pK^{t,\mathcal{O}}_{(b-1)/a}(x)> s\mright\}
        \end{aligned}
        \right..
    \]
\end{definition}

\begin{definition}[$\mathsf{Gap}\textnormal{-}\MINpKT$]
        For an oracle $\mathcal{O}$ and $\rho\colon\mathbb{N} \to \mathbb{N}$, we define the promise language
        \[
\mathsf{Gap}_{\rho}\textnormal{-}\MINpKT^{\mathcal{O}}\eqdef (\mathsf{YES}, \mathsf{NO})
        \]
        as
    \[
        \left\{
        \begin{aligned}
            \mathsf{YES} & = \mleft\{(x,1^s,1^t,1^a,1^b) \mid \pK^{t, \mathcal{O}}_{b/a}(x)\leq s\mright\} \\
            \mathsf{NO} & = \mleft\{(x,1^s,1^t,1^a,1^b) \mid \pK^{\rho(t+a+b), \mathcal{O}}_{(b-1)/a}(x)\geq s + \log\rho(t+a+b)\mright\}
        \end{aligned}
        \right..
    \]
    For a complexity class $\mathcal{C}$, we say that $\mathsf{Gap}\textnormal{-}\mathsf{MINpKT}^{\mathcal{O}} \in \mathcal{C}$ if there exists a constant $c > 0$ such that $\mathsf{Gap}_\rho\textnormal{-}\mathsf{MINpKT}^{\mathcal{O}} \in \mathcal{C}$ with $\rho(n) = n^c$.%
\end{definition}

\begin{definition}[$\mathsf{Mild}\textnormal{-}\mathsf{Gap}\textnormal{-}\MINpKT$]
        For an oracle $\mathcal{O}$ and a function $\rho\colon\mathbb{N}\to \mathbb{N}$, we define the promise language $\mathsf{Mild}\textnormal{-}\mathsf{Gap}_{\rho}\textnormal{-}\MINpKT^{\mathcal{O}}\eqdef (\mathsf{YES}, \mathsf{NO})$ as
    \[
        \left\{
        \begin{aligned}
            \mathsf{YES} & = \mleft\{(x,1^s,1^t,1^a,1^b) \mid \pK^{t, \mathcal{O}}_{b/a}(x)\leq s \mright\} \\
            \mathsf{NO} & = \mleft\{(x,1^s,1^t,1^a,1^b) \mid \pK^{t + \rho(|x|+a+b), \mathcal{O}}_{(b-1)/a}(x)\geq s + \log\rho(t+a+b)\mright\}
        \end{aligned}
        \right..
    \]
    For a complexity class $\mathcal{C}$, we say that $\mathsf{Mild}\textnormal{-}\mathsf{Gap}\textnormal{-}\mathsf{MINpKT}^{\mathcal{O}} \in \mathcal{C}$ if there exists a constant $c > 0$ such that $\mathsf{Mild}\textnormal{-}\mathsf{Gap}_\rho\textnormal{-}\mathsf{MINpKT}^{\mathcal{O}} \in \mathcal{C}$ with $\rho(n) = n^c$.
\end{definition}

    For all the promise problems 
    \[\Pi \in \{\mathsf{Gap}_{\rho}\textnormal{-}\MINKT, \mathsf{Mild}\textnormal{-}\mathsf{Gap}_{\rho}\textnormal{-}\MINKT, \mathsf{Gap}_{\rho}\textnormal{-}\MINpKT, \mathsf{Mild}\textnormal{-}\mathsf{Gap}_{\rho}\textnormal{-}\MINpKT\},\] we adopt the following convention. For complexity classes $\mathcal{C}$ and $\mathcal{D}$, we say that $\Pi^{\mathcal{D}} \in \mathcal{C}$ if, for every language $L \in \mathcal{D}$, there exists a constant $c > 0$ such that $\Pi^{L} \in \mathcal{C}$ with gap parameter $\rho(n) = n^c$. In particular, note that $\rho$ may depend on the choice of $L$.

\subsection{Useful Tools}
Here, we state some useful tools about Kolmogorov complexity.

\begin{lemma}[Standard Counting Argument]\label{l:incomp-strings}
    For every $t\in\mathbb{N}$, $\mu \in \{\K, \K^t, \nK^t\}$, and every $n, \gamma \in \mathbb{N}$, we have
    \[
        \Prob_{x \sim \{0,1\}^n} \left[ \mu(x) < n - \gamma \right] < \frac{1}{2^\gamma}.
    \]
    Also, for every $t\in\mathbb{N}, \lambda \in (0, 1)$, $\mu \in \{\pK_{\lambda}^t, \pnK_{\lambda}^t\}$, and every $n, \gamma \in \mathbb{N}$, we have
    \[
        \Prob_{x \sim \{0,1\}^n} \left[ \mu(x) < n - \gamma \right] < \frac{1}{\lambda \cdot 2^\gamma}.
    \]
    Moreover, the inequalities above continue to hold in the presence of an oracle $\mathcal{O}$ and when conditioning on a fixed string $y$.
\end{lemma}

\begin{proof}

    For simplicity, we consider the oracle-free case and $y = \epsilon$. (The general case is analogous.) The case of $\K$ directly follows from counting. Indeed, the number of programs with length less than $n - \gamma$ is at most $2^{n - \gamma}-1$. Therefore,
    \[
        \Prob_{x \sim \{0, 1\}^n}\left[\K(x) < n - \gamma\right] < \frac{2^{n - \gamma}}{2^n} = \frac{1}{2^{\gamma}}.
    \]

    The deterministic and nondeterministic time bounded cases follow similarly, since each program of length less than $n-\gamma$ contributes at most one output string.

    For $\mu = \pK^t_\lambda, \pnK^t_\lambda$, consider a bipartite graph where the left vertices represent randomness-program pairs  $\{0,1\}^t \times \{0,1\}^{<n - \gamma}$, and the right vertices represent strings in $\{0,1\}^n$. There is an edge between a left vertex $(r, M)$ and a right vertex $x$ if $M(r)$ outputs $x$. For a right vertex $x$ to have $\pK_{\lambda}^t$-complexity at most $n - \gamma$, it must have degree at least $2^t \cdot \lambda$. 
    
    However, there are less than $2^t \cdot 2^{n - \gamma}$ total edges in the graph. Therefore, the number of right vertices with degree at least $2^t \cdot \lambda$ is less than $2^{n-\gamma}/\lambda$.
\end{proof}

	\begin{lemma}[Following {\cite[Lemma 18]{GKLO22}}]\label{l:K-pK}
	Let $L$ be any computable language. There exists a constant $c>0$ such that for any $x\in\bool^*$, time bound $t\in\mathbb{N}$, $\K(x\mid t) \leq \pK^{t,L}(x) + c\cdot \log |x|$.
\end{lemma}

\begin{theorem}[Language Compression for $\mathsf{pnK}^t$ \cite{DBLP:journals/eccc/HuL025}]\label{t:lc-pnK}
	There is a polynomial $p_{_\mathsf{LC}}(\cdot)$ such that for every set $A\vcentcolon=\{A_n\subseteq \bool^n\}_{n\in\mathbb{N}}$, and for all $n\in\mathbb{N}$ and $x\in A_n$,
	\[
	\mathsf{pnK}^{p_{_\mathsf{LC}(n)}, A_n}(x) \leq \log |A_n| + \log p_{_\mathsf{LC}}(n).
	\]
\end{theorem}

\paragraph{Direct Product Generator.}
For $x,y\in\{0,1\}^n$, we denote their inner product by $x\cdot y=\bigoplus_{i=1}^nx_iy_i$.
\begin{definition}[Direct Product Generator (DPG) \cite{DBLP:conf/stoc/Hirahara21}]\label{def:DPG}
    For $n, k \in \mathbb{N}$, the $k$-wise direct product generator is the function
    \[
        \mathsf{DP}_k \colon \{0,1\}^n \times \{0,1\}^{nk} \to \{0,1\}^{nk + k},
    \]
    defined by
    \[
        \mathsf{DP}_k(x, z_1 \circ \cdots \circ z_k) = (z_1 \circ \cdots \circ z_k \circ x \cdot z_1 \circ \cdots \circ x \cdot z_k).
    \]
\end{definition}
\begin{lemma}[$\pK^t$ Reconstruction Lemma \cite{GKLO22}]\label{lem:recon}
		There is a polynomial $p_{_\mathsf{DP}}$ such that the following holds. For every $n,m,k\in\mathbb{N}^+$, $x \in \bool^n$, let $D$ be a function that $(1/m)$-distinguishes $\mathsf{DP}_k(x; \mathcal{U}_{nk})$ from $\mathcal{U}_{nk + k}$. Then  
		\[\pK^{p_{_\mathsf{DP}}(n+m+k), \parallel D}(x) \leq k +\log p_{_\mathsf{DP}}(n+m+k).\]
	\end{lemma}

\paragraph{A SAT Oracle Simulation Lemma.}

\begin{lemma}\label{lem:kA-leq-nk-super-script}
    There exists a constant $c > 0$ such that for any strings $x, y \in \bool^*$ and $t \in \mathbb{N}$,
    \[
        \K^{t', \parallel \SAT}(x \mid y) \leq \nK^t(x \mid y) + c\cdot \log t,
    \]
    where $t' = (t + |y|)^c$.
\end{lemma}

\begin{proof}
    Let $x, y \in \bool^*$ and $t \in \mathbb{N}$. We assume without loss of generality that $t \geq |x|$. Define $k \vcentcolon= \nK^t(x \mid y)$, and let $\Pi \in \bool^k$ be a program that nondeterministically outputs $x$ on input $y$ in time $t$. 

    Given $\Pi$, $|x|$, $t$, and $y$, we reconstruct $x$ as follows: For each $i \in [|x|]$, we use the $\SAT$ oracle to answer the following query: 
    \begin{quote}
        Does there exist a $w \in \bool^t$ such that $\Pi^{\mathcal{O}}(y; w)$ outputs a string $z \in \bool^{|x|}$ in time $t$ and the $i$-th bit of $z$ is $1$?
    \end{quote}
    We then set the $i$-th bit of $x$ to be $1$ if and only if the answer to this query is yes.    

    It is easy to verify that the above procedure runs in time $\poly(|x| + |y| + t)$. Also, by the fact that $\Pi$ nondeterministically outputs $x$ on input $y$ in time $t$, it can correctly recover $x$. Thus, for $t' = \poly(|y| + t)$, we obtain
    \begin{align*}
        \K^{t', \parallel \SAT}(x \mid y) &\leq |\Pi| + O(\log |t|) \\
        &\leq k + O(\log t),
    \end{align*}
    as desired.
\end{proof}

\begin{lemma}\label{l:pKt-SAT-superscript}
        For every oracle $\mathcal{O}$ and every language $L$ that is hard for $\NP^{\mathcal{O}}$ under deterministic polynomial-time reductions, there exists a constant $c  > 0$ such that, for all $x,y\in\bool^*$ and $t\in\mathbb{N}$,
        \[
            \pK^{t^c,L}(x\mid y)\leq\pnK^{t, \mathcal{O}}(x\mid y) + c\cdot \log t.
        \]
    \end{lemma}
    \begin{proof}

        The proof follows the same idea as \Cref{lem:kA-leq-nk-super-script}. Let $k = \pnK^{t,\mathcal{O}}(x \mid y)$, and suppose that with probability at least $2/3$ over $r \sim \bool^t$, there exist $\Pi \in \bool^k$ such that, on input $(y;r)$ and with oracle access to $\mathcal{O}$, $\Pi$ non-deterministically outputs $x$ within time $t$.

        Let $L_0 \in \NP^{\mathcal{O}}$ be the following language:
        \[L_0 = \mleft\{(\Pi,y,r, i, 1^t, 1^n) \,\,\middle\vert\,\, 	
        	\begin{aligned}
		&\exists w\in\bool^t \textrm{ such that
		$\Pi^{\mathcal{O}}(y;w;r)$ outputs a string}\\
        &x \in \bool^n \textrm{ within time $t$ and the $i$-th bit of $x$ is $1$.}
	\end{aligned}
    \mright\}.\]

        Then, with probability at least $2/3$ over $r \sim \bool^t$, given $(\Pi, y, r, t, n)$ as above hard-coded, one can recover the string $x$ within time $\poly(t)$ given oracle access to $L_0$. In particular, one can simply query $(\Pi,y,r, i, 1^t, 1^n)$ to $L_0$ for every $i \in [n]$ and return the concatenation of the answers from the oracle. By definition of $\Pi$ and $t$, this process is guaranteed to output $x$. 
        
        Lastly, for any $\NP^{\mathcal{O}}$-complete language $L$, we may replace the $L_0$-oracle with an $L$-oracle by simulating the reduction from $L_0$ to $L$. This increases the time to produce $x$ by at most a factor of $\poly(t)$. 

        Overall, we have
        \begin{align*}
            \pK^{\poly(t), L}(x \mid y) &\leq |(\Pi, t, n)| + O(\log t)
            \\&\leq  |\Pi| + O(\log t)
            \\&= \pnK^{t,\mathcal{O}}(x \mid y)) + O(\log t),
        \end{align*}
        as desired.
    \end{proof}

     \subsection{Average-Case Complexity}\label{ss:ave_case_def}
     \paragraph{Basic Definitions and Facts.}
We review some basic definitions and facts in average-case complexity. For more information about average-case complexity, we refer to \cite{DBLP:journals/fttcs/BogdanovT06}.
     
Recall that a pair $(L, \mathcal{D})$ is a \emph{distributional problem} if $L \subseteq \{0,1\}^*$ and $\mathcal{D} = \{\mathcal{D}_n\}_{n \in \mathbb{N}}$ is a distribution family, where each $\mathcal{D}_n$ is a distribution over $\{0,1\}^*$.

We let $\mathsf{DistNP}$ denote the set of distributional problems $(L, \mathcal{D})$ such that $L \in \mathsf{NP}$ and $\mathcal{D}$ is polynomial-time samplable. Here, a distribution family $\mathcal{D} = \{\mathcal{D}_n\}_{n \in \mathbb{N}}$ is \emph{polynomial-time samplable} if there exists a randomized polynomial-time algorithm $A$ such that for every $n \geq 1$, the output of $A(1^n)$ is distributed according to $\mathcal{D}_n$.  We let $\mathsf{PSAMP}$ be the collection of distribution families that can be sampled in polynomial time.
 For a complexity class $\mathcal{C}$ (e.g., $\mathcal{C}$ = NP), we let $\Dist\mathcal{C}$ denote the set of distributional problems $(L, \mathcal{D})$
with $L \in \mathcal{C}$ and $\mathcal{D} \in \mathsf{PSAMP}$.

A distributional problem $(L,\mathcal{D})$ is said to admit an \emph{errorless heuristic scheme} if there exists a deterministic polynomial-time algorithm $A$ such that, for every $n,k\in\mathbb{N}$,
\begin{enumerate}
    \item for every $x\in\supp(\mathcal{D}_n)$,
    \[
    A(x,1^n,1^k)\in\{L(x),\bot\},
    \]
    \item and
    \[
    \Prob_{x\sim\mathcal{D}_n}\!\left[A(x,1^n,1^k)\neq\bot\right]
    \ge 1-\frac1k.
    \]
\end{enumerate}
We let $\mathsf{AvgP}$ denote the class of distributional problems that admit an errorless heuristic scheme.

The above definition naturally extends to randomized algorithms. A distributional problem $(L,\mathcal{D})$ is said to admit a \emph{randomized errorless heuristic scheme} if there exists a probabilistic polynomial-time algorithm $A$ such that, for every $n,k\in\mathbb{N}$,
\begin{enumerate}
    \item for every $x\in\supp(\mathcal{D}_n)$,
    \[
    \Prob_A\!\left[A(x,1^n,1^k)\in\{L(x),\bot\}\right]\ge\frac45,
    \]
    \item and
    \[
    \Prob_{x\sim\mathcal{D}_n}\!\left[
    \Prob_A\!\left[A(x,1^n,1^k)=\bot\right]
    <\frac15
    \right]
    \ge 1-\frac1k.
    \]
\end{enumerate}
We let $\mathsf{AvgBPP}$ denote the class of distributional problems that admit randomized errorless heuristic schemes.

The above definitions require the heuristic algorithm to be errorless, allowing it only to output $\bot$ instead of an incorrect answer. We can relax this requirement by permitting the algorithm to err on both yes- and no-instances.

A distributional problem $(L,\mathcal{D})$ is said to admit a \emph{randomized two-sided-error heuristic scheme} if there exists a probabilistic polynomial-time algorithm $A$ such that, for every $n,k\in\mathbb{N}$,
\[
\Prob_{x\sim\mathcal{D}_n}\!\left[
\Prob_A\!\left[A(x,1^n,1^k)=L(x)\right]
>\frac35
\right]
\ge 1-\frac1k.
\]
We let $\mathsf{HeurBPP}$ denote the class of distributional problems that admit randomized two-sided-error heuristic schemes.

We will need the following result.
\begin{theorem}[\cite{DBLP:conf/focs/ImpagliazzoL90}]\label{t:IL90}
    Let $A$ be any oracle. If $(\NP^{A},\mathcal{U})\subseteq\Avg\BPP$, then for every distribution $\mathcal{D} \in \mathsf{PSAMP}$ and $L \in \NP^A$, we have $(L,\mathcal{D}) \in \mathsf{AvgBPP}$.
\end{theorem}

\paragraph{Notions of Average-Case Tractability for Meta-Complexity Problems.}

	Let $A$ be an algorithm, $\Pi\vcentcolon=(\mathsf{YES}, \mathsf{NO})$ be a promise problem, and $z\in\bool^*$, we write ``$A$ decides $\Pi$ on $z$'' if the following holds.
	\[
	A(z) = \begin{cases}
		1  & \text{if }z\in \mathsf{YES}, \\
		0 & \text{if } z\in \mathsf{NO},\\
		\textbf{either $0$ or $1$} & \text{otherwise.}
	\end{cases}
	\]
	
	\begin{definition}[``$\mathsf{Gap}\textnormal{-}\mathsf{MINpKT}^{\mathsf{PH}}\in \mathsf{Heur}\BPP$'']\label{def:two-sided}
    Let $O$ be an oracle and let $\rho\colon\mathbb{N} \to \mathbb{N}$. We say that $\mathsf{Gap}_{\rho}\textnormal{-}\mathsf{MINpKT}^{O}$ is easy on average with \emph{two-sided error} if for every polynomial-time samplable distribution family $\{\mathcal{D}_n\}_n$ supported over $\bool^n$, there exist a PPT algorithm $A$ such that for every $n,s,t,a,b,k\in\mathbb{N}$,
\[
\Prob_{x\sim\mathcal{D}_n}\!\left[
\Prob_A\mleft[A(-;1^k) \textnormal{ decides } \mathsf{Gap}_{\rho}\textnormal{-}\mathsf{MINpKT}^{O} \textnormal{ on } (x,1^s,1^t,1^a,1^b)\mright]
>\frac35
\right]
\ge 1-\frac1k.
\]
    We say that $\mathsf{Gap}\textnormal{-}\mathsf{MINpKT}^{\mathsf{PH}}$ is easy on average with two-sided error, denoted as ``$\mathsf{Gap}\textnormal{-}\mathsf{MINpKT}^{\mathsf{PH}}\in \mathsf{Heur}\BPP$'', if for every $\mathcal{O}\in \mathsf{PH}$, there exist a polynomial $\rho$ such that $\mathsf{Gap}_{\rho}\textnormal{-}\mathsf{MINpKT}^{\mathcal{O}}$ is easy on average with two-sided error.%
	\end{definition}

	In the following zero error setting, we further require that the algorithm will never output an incorrect answer on both the positive and negative instances (except with small probability over the internal randomness of the algorithm), while it can output $\bot$ to indicate that it does not know the correct answer.
	
	\begin{definition}[``$\mathsf{Gap}\textnormal{-}\mathsf{MINpKT}^{\mathsf{PH}}\in \mathsf{Avg}\BPP$'']\label{def:zero-sided}
		   Let $O$ be an oracle and let $\rho\colon\mathbb{N} \to \mathbb{N}$. We say that $\mathsf{Gap}_{\rho}\textnormal{-}\mathsf{MINpKT}^{O}$ is easy on average with \emph{zero error} if for every polynomial-time samplable distribution family $\{\mathcal{D}_n\}_n$ supported over $\{0,1\}^n$, there exist a PPT algorithm $A$ such that for all $n, s, t, a, b, k\in\mathbb{N}$,
\[
\Prob_{x\sim\mathcal{D}_n}\!\left[
\Prob_A\mleft[A(-;1^k) \textnormal{ decides } \mathsf{Gap}_{\rho}\textnormal{-}\mathsf{MINpKT}^{O} \textnormal{ on } (x,1^s,1^t,1^a,1^b)\mright]
>\frac35
\right]
\ge 1-\frac1k.
\]and for every $x\in\supp(\mathcal{D}_n$,
		\[
		\Prob_{A} \mleft[A(-;1^k) \textnormal{ decides } \mathsf{Gap}_{\rho}\textnormal{-}\mathsf{MINpKT}^{O} \textnormal{ on } (x,1^s,1^t,1^a,1^b) \textnormal{\textbf{ OR }} A(x,1^s,1^t,1^a,1^b;1^k)=\bot\mright] \geq \frac{4}{5}.
		\]
    We say that $\mathsf{Gap}\textnormal{-}\mathsf{MINpKT}^{\mathsf{PH}}$ is easy on average with zero error, denoted as ``$\mathsf{Gap}\textnormal{-}\mathsf{MINpKT}^{\mathsf{PH}}\in \mathsf{Avg}\BPP$'', if for every $\mathcal{O}\in \mathsf{PH}$, there exist a polynomial $\rho$ such that $\mathsf{Gap}_{\rho}\textnormal{-}\mathsf{MINpKT}^{\mathcal{O}}$ is easy on average with zero error.%
	\end{definition} 
	
    It is straightforward to extend the above definitions to $\mathsf{Mild}\textnormal{-} \mathsf{Gap}\textnormal{-}\mathsf{MINpKT}^{\mathsf{PH}}$ and $\mathsf{MINpKT}^{\mathsf{PH}}$.

\section{Proof of the Equivalence Theorem}\label{sec:proof_equiv}

In this section, we prove \Cref{t:PH-heuristica-characterization}, which is restated below.

\PHHEUR*

\Cref{t:PH-heuristica-characterization} follows directly from the following two theorems using that $\NP \subseteq \BPP$ implies $\PH \subseteq \BPP$.

\begin{theorem}\label{t:MINpnKT-NP}
    		The following holds.
		\[
		\NP\subseteq\BPP \iff \mathsf{Mild}\textnormal{-} \mathsf{Gap}\textnormal{-}\mathsf{MINpKT}^{\mathsf{PH}} \in \prBPP.
		\]
\end{theorem}

\begin{restatable}{theorem}{PHMINpKT}\label{t:GapMINpnKT-PH}
        		The following holds.
		\[
         \Dist\mathsf{PH}\subseteq\Avg\BPP \iff
		\mathsf{Gap}\textnormal{-}\mathsf{MINpKT}^{\mathsf{PH}} \in \prBPP.
		\]
\end{restatable}

In the remainder of this section, we show \Cref{t:MINpnKT-NP} and \Cref{t:GapMINpnKT-PH}.
We prove \Cref{t:MINpnKT-NP} in \Cref{sec:NP-hardness-MINpKT-PH}.
The forward direction of \Cref{t:GapMINpnKT-PH} is established by \Cref{l:PH-to-GapMINpnKT}, which is proved in \Cref{sec:PH-to-GapMINpnKT}. The backward direction follows from \Cref{l:GapMINpnKT-to-PH}, with the proof provided in \Cref{sec:GapMINpnKT-to-PH}.

\subsection{Characterizing Easiness of \texorpdfstring{$\NP$}{} via Easiness of \texorpdfstring{$\mathsf{Mild}\textnormal{-}\mathsf{Gap}\textnormal{-}\mathsf{MINpKT}^{\mathsf{PH}}$}{}}\label{sec:NP-hardness-MINpKT-PH}
In this section, we prove \Cref{t:MINpnKT-NP}. The theorem follows directly from \Cref{t:MINpnKT-NP-easy} and \Cref{t:MINpnKT-NP-hard}, which are proved in the subsections below.

\subsubsection{The Easy Direction}

\begin{lemma}\label{t:MINpnKT-NP-easy}
	If $\NP\subseteq\BPP$, then $\mathsf{Mild}\textnormal{-} \mathsf{Gap}\textnormal{-}\mathsf{MINpKT}^{\mathsf{PH}} \in \prBPP$.
\end{lemma}
\begin{proof}
It is easy to see that, for any oracle $\mathcal{O}$, we have $\mathsf{Mild}\textnormal{-} \mathsf{Gap}\textnormal{-}\mathsf{MINpKT}^{\mathcal{O}} \in \mathsf{prAM}^{\mathcal{O}}$.\footnote{This is true even for $\mathsf{MINpKT}^{\mathcal{O}}$.} Also, if $\NP \subseteq \BPP$, then $\PH \subseteq \BPP$. It follows that $\mathsf{prAM}^{\PH} \subseteq \mathsf{prAM}^{\BPP} = \mathsf{prAM} \subseteq \mathsf{prBPP}^{\NP} \subseteq  \mathsf{prBPP}^{\BPP} = \mathsf{prBPP}$.
\end{proof}

\subsubsection{The Hard Direction}
\begin{lemma}\label{t:MINpnKT-NP-hard}
    If $\mathsf{Mild}\textnormal{-} \mathsf{Gap}\textnormal{-}\mathsf{MINpKT}^{\mathsf{\SAT}} \in \prBPP$, then $\NP\subseteq\BPP$.
\end{lemma}

The proof of \Cref{t:MINpnKT-NP-hard} follows the approach in \cite{DBLP:journals/eccc/HuL025}, which shows that if $\mathsf{MINnKT} \in \BPP$, then $\mathsf{NP} \subseteq \BPP$. Here, we sketch the proof and highlight the key steps and modifications.

We begin by describing the approach at a high level. Let $L \in \NP$, and let $x \in \bool^n$ be an input instance. Our goal is to find a witness for $x$. To this end, suppose $x \in L$, and let $A_x$ denote the set of $L$-witnesses for $x$. We aim to establish the following: for every $y \in A_x$,
\begin{equation}\label{eq:pK-sampling-witness}
\pK^{\poly(n)}(y \mid x) \leq \log |A_x| + O(\log n).
\end{equation}
The above is a consequence of a result known as \emph{language compression} (for $\pK^t$), which states that every witness $y \in A_x$ can be compressed to approximately $\log |A_x|$ bits.

Note that if $\pK^{\poly(n)}(y \mid x) \leq k$, then by the definition of $\pK^t$, if we randomly sample a string $r \sim \bool^{\poly(n)}$ and a program $\Pi \sim \bool^{\leq k}$, and output $\Pi(x; r)$, we obtain $y$ with probability at least $1/(2^k \cdot \poly(n))$. Then, by \Cref{eq:pK-sampling-witness}, for every $y \in A_x$, the above procedure samples $y$ with probability at least $1/(|A_x| \cdot \poly(n))$. Since this holds for each of the $|A_x|$ elements in $A_x$, it follows that we obtain some string in $A_x$ with probability at least $1/\poly(n)$. By repeating this procedure, we can obtain a witness for $x$ with high probability.

It is not known whether $\pK^t$ admits a language compression property as stated in \Cref{eq:pK-sampling-witness}. However, it was shown in \cite{DBLP:journals/eccc/HuL025} that its nondeterministic variant, $\pnK^t$, indeed admits such a property. Since $\pK^{t,\SAT}$ is a more powerful notion than $\pnK^t$, it follows that $\pK^{t,\SAT}$ also admits this language compression property. In particular, we get that for every witness $y\in A_x$,
\begin{align*}
\pK^{\poly(n),\SAT}(y \mid x) &\leq \pnK^{\poly(n)}(y \mid x) + O(\log n) \tag{by \Cref{l:pKt-SAT-superscript}}\\
&\leq \pnK^{\poly(n), A_x}(y) + O(\log n)\\
&\leq \log |A_x| + O(\log n). \tag{by \Cref{t:lc-pnK}}  
\end{align*}

Thus, it suffices to show that if $\mathsf{Mild}\textnormal{-}\mathsf{Gap}\textnormal{-}\mathsf{MINpKT}^{\mathsf{SAT}} \in \prBPP$, then for every $x, y \in \bool^n$ and every $t$ that is a sufficiently large polynomial in $n$,
\begin{equation}\label{eq:pK-pnK}
\pK^{\poly(t)}(y \mid x) \leq \pK^{\poly(n), \SAT}(y \mid x) + O(\log t).
\end{equation}

We first note that, following previous work, one can establish a weaker form of \Cref{eq:pK-pnK}, which states that
\begin{equation}\label{eq:pK-pnK-cd}
\pK^{\poly(t)}(y \mid x) \leq \pK^{\poly(n), \SAT}(y \mid x) + \pK^{t,\SAT}(x) - \pK^{\poly(t),\SAT}(x) + O(\log t).
\end{equation}
Note that this weaker bound is not sufficient for our purposes, since the term $\pK^{t,\SAT}(x) - \pK^{\poly(t),\SAT}(x)$, known as the (time-bounded) computational depth, can be large for some $x$. Nevertheless, we sketch how \Cref{eq:pK-pnK-cd} can be shown under the assumption that $\mathsf{Mild}\textnormal{-}\mathsf{Gap}\textnormal{-}\mathsf{MINpKT}^{\mathsf{SAT}}$ is easy.

The idea is to show both an upper bound and a lower bound for the quantity
\[
    \pK^{t + \poly(n), \SAT}(x, y).
\]
Roughly speaking, we aim to show that:
\begin{itemize}
    \item \textbf{(Upper bound)} $\pK^{t + \poly(n), \SAT}(x, y) \lesssim \pK^{\poly(n), \SAT}(y \mid x) + \pK^{t,\SAT}(x)$.
    \item \textbf{(Lower bound)} $\pK^{t + \poly(n), \SAT}(x, y) \gtrsim \pK^{\poly(t)}(y \mid x) + \pK^{\poly(t), \SAT}(x)$.
\end{itemize}

It is easy to combine the two inequalities above to obtain \Cref{eq:pK-pnK-cd}.

The upper bound is straightforward. 
The lower bound can be established under the assumption that $\mathsf{Mild}\textnormal{-}\mathsf{Gap}\textnormal{-}\mathsf{MINpKT}^{\mathsf{SAT}}$ is easy. For simplicity, let us assume that we have a deterministic polynomial-time algorithm $A$ that given $(z,1^s,1^t)$ decides whether $\pK^{t,\SAT}(z)\leq s$.

Let $\mathsf{DP}_k \colon \bool^n \times \bool^{nk} \to \bool^{nk+k}$ be the direct product generator defined in \Cref{def:DPG}, where
\[
    k \vcentcolon= \pK^{t^c}(y \mid x) - c \cdot \log t,
\]
for some large constant $c > 1$ to be specified later. We then consider the following two distributions:
\begin{align*}
    \mathcal{D}_1 &\vcentcolon= \mleft(x, \mathsf{DP}_k(y, \mathcal{U}_{nk})\mright), \\
    \mathcal{D}_2 &\vcentcolon= \mleft(x, \mathcal{U}_{nk+k}\mright).
\end{align*}
The goal is to show that $\mathcal{D}_1$ and  $\mathcal{D}_2$ are indistinguishable under the algorithm $A(-, 1^s, 1^{t+n^d})$ for some $s$ and constant $d$ to be specified later. The reason for doing this is that, assuming $\mathsf{Mild}\textnormal{-}\mathsf{Gap}\textnormal{-}\mathsf{MINpKT}^{\mathsf{SAT}}$ is easy, one can show a weak time-bounded symmetry of information \cite{DBLP:conf/stoc/Hirahara21} (see also \cite{DBLP:conf/focs/Ilango23}), which states that, with high probability over $w \sim \bool^{nk+k}$, we have
\[
    \pK^{t+n^d, \SAT}(x, w) \geq \pK^{(t+n^d)^{b}, \SAT}(x) + nk + k - b \log t,
\]
where $b\geq 0$ is a constant.
Therefore, with high probability (over $\mathcal{D}_2$), $A(\mathcal{D}_2, 1^s, 1^{t+n^d}) = 0$
for $s \vcentcolon= \pK^{(t+n^d)^{b}, \SAT}(x) + nk + k - b\log t$.  

Now if $\mathcal{D}_1$ is indeed indistinguishable from $\mathcal{D}_2$, then with high probability (over $\mathcal{D}_1$), we also have $A(\mathcal{D}_1, 1^s, 1^{t+n^d}) = 0$,
which implies that there exists some $z \in \bool^{nk}$ such that
\[
    \pK^{t+n^d, \SAT}\mleft(x, \mathsf{DP}_k(y, z)\mright) \geq s.
\]
Also, since $\mathsf{DP}_k(y, z)$ is efficiently computable given $y$ and $z$, by letting $d$ be a sufficiently large constant, we have
\[
    \pK^{t+n^d, \SAT}\mleft(x, \mathsf{DP}_k(y, z)\mright) \leq \pK^{t+\poly(n), \SAT}(x, y) + nk + O(\log t).
\]
Combining the above two inequalities, we get
\begin{align*}
    \pK^{t+\poly(n), \SAT}(x, y) &\geq s - nk - O(\log t) \\
    &= \pK^{\poly(t), \SAT}(x) + k - O(\log t) \\
    &= \pK^{\poly(t), \SAT}(x) + \pK^{\poly(t)}(y \mid x) - O(\log t),
\end{align*}
which gives the desired lower bound.

Now, let us explain why $\mathcal{D}_1$ and $\mathcal{D}_2$ are indistinguishable. Indeed, if $A(-, 1^s, 1^{t+n^d})$ could distinguish $\mathcal{D}_1$ from $\mathcal{D}_2$, then by the fact that $A$ is polynomial-time computable and by the reconstruction property of $\mathsf{DP}_k$ (\Cref{lem:recon}), we would obtain
\[
    \pK^{\poly(t)}(y \mid x) \leq k + O(\log t),
\]
which contradicts our choice of $k$ be letting $c$ be a sufficiently large constant.

Next, let us return to \Cref{eq:pK-pnK-cd}. It turns out that if we could show a more fine-grained form:
\begin{equation}\label{eq:intro_techniques_ideal}
    \pK^{\poly(t)}(y \mid x) \leq \pK^{\poly(n),\SAT}(y \mid x) + \pK^{t,\SAT}(x) - \pK^{t + p(n), \SAT}(x) + O(\log n),
\end{equation}
where $p(n)$ is a polynomial independent of $t$, then we could aim to conclude the proof by showing that
\[
    \pK^{t, \SAT}(x) - \pK^{t + p(n), \SAT}(x) = O(\log n).
\]
While this may not hold for a particular value of $t$, one can argue, via a telescoping sum trick, that such a bound holds for some $t'$ that is not too large. More concretely, for any fixed time bound $t$, let $t_i = t + i \cdot p(n)$ for $i \in \{0,1, \ldots, n\}$. Note that $\pK^{t_0, \SAT}(x) - \pK^{t_n, \SAT}(x) \leq 2n$,
since any $n$-bit string $x$ can be described using at most $2n$ bits. Therefore,
\[
    \sum_{i=0}^{n-1} \big(\pK^{t_i, \SAT}(x) - \pK^{t_{i+1}, \SAT}(x)\big) = \pK^{t_0, \SAT}(x) - \pK^{t_n, \SAT}(x) \leq 2n,
\]
which implies that there exists an index $i$ such that $\pK^{t_i, \SAT}(x) - \pK^{t_{i+1}, \SAT}(x) \leq 2$.
In particular, for some time bound $t' \leq t + n \cdot p(n)$, we have $\pK^{t', \SAT}(x) - \pK^{t' + p(n), \SAT}(x) = O(\log n)$,
which would be sufficient for our purpose.

Unfortunately, it is not clear how to prove \Cref{eq:intro_techniques_ideal} exactly. As illustrated earlier, a key difficulty lies in the use of weak symmetry of information, which introduces a time overhead of the form $\poly(t)$ rather than an \emph{additive} overhead of the form $t + p(n)$.

Here, we follow the approach developed in \cite{DBLP:journals/eccc/HuL025} which incorporates the ideas described above and suffices to prove \Cref{t:MINpnKT-NP-hard}. A key ingredient is the following new notion of time-bounded Kolmogorov complexity.

\begin{definition}[$\ell\lpK_{\lambda,\gamma}^t$]\label{d:l-pnK}
    For $t, \ell \in \mathbb{N}$, $\lambda, \gamma \in (0,1)$, an oracle $\mathcal{O}$, and a string $x \in \{0,1\}^*$, we define
    \[
        \ell\lpK_{\lambda,\gamma}^{t,\mathcal{O}}(x) \eqdef \min\mleft\{ s \in \mathbb{Z} \;\middle\vert\; \Prob_{\hat r \sim \mathcal{U}_\ell} \mleft[\pK_\lambda^{t,\mathcal{O}}(x \circ \hat r) \le s + \ell \mright] \ge \gamma \mright\}.
    \]
\end{definition}

It can be shown that $\ell\lpK_{\lambda,\gamma}^t$ is a well-defined measure that satisfies several basic properties one would expect from a Kolmogorov complexity measure; see \cite[Section 4.2.1]{DBLP:journals/eccc/HuL025}. Moreover, if the mild-gap version of approximating $\pK^{t}$ is easy, then so is approximating $\ell\lpK_{\lambda,\gamma}^t$.

\begin{definition}[$\mathsf{Mild}\textnormal{-}\mathsf{Gap}\textnormal{-}\MINlpKT$]
    For an oracle $\mathcal{O}$ and $\rho\colon \mathbb{N} \to \mathbb{N}$, we define the promise language $\mathsf{Mild}\textnormal{-}\mathsf{Gap}_{\rho}\textnormal{-}\MINlpKT^{\mathcal{O}} \eqdef (\mathsf{YES}, \mathsf{NO})$ as
    \begin{equation*}
        \left\{
        \begin{aligned}
            \mathsf{YES} & = \mleft\{(x,1^s,1^t,1^{\ell},1^a,1^b,1^c,1^d) \mid \ell\lpK^{t, \mathcal{O}}_{b/a,d/c}(x) \leq s \mright\}, \\
            \mathsf{NO} & = \mleft\{(x,1^s,1^t,1^\ell,1^a,1^b,1^c,1^d) \mid \ell\lpK^{t + \rho(|x|+\ell+a+b+c+d), \mathcal{O}}_{(b-1)/a, (d-1)/c}(x) \geq s + \log \rho(t+\ell+a+b+c+d) \mright\}.
        \end{aligned}
        \right.
    \end{equation*}
    We say that $\mathsf{Mild}\textnormal{-}\mathsf{Gap}\textnormal{-}\MINlpKT^{\mathcal{O}} \in \prBPP$ if there is a constant $g > 0$ such that $\mathsf{Mild}\textnormal{-}\mathsf{Gap}_\rho\textnormal{-}\MINlpKT^{\mathcal{O}} \in \prBPP$ with $\rho(n) = n^g$.
\end{definition}

\begin{lemma}\label{l:MINlpnKT-from-MINpnKT}
    For every oracle $\mathcal{O}$, if $\mathsf{Mild}\textnormal{-}\mathsf{Gap}\textnormal{-}\MINpKT^{\mathcal{O}} \in \prBPP$, then $\mathsf{Mild}\textnormal{-}\mathsf{Gap}\textnormal{-}\MINlpKT^{\mathcal{O}} \in \prBPP$.
\end{lemma}

\begin{proof}
    The proof is adapted from that of \cite[Lemma 69]{DBLP:journals/eccc/HuL025}, which shows a similar result for the case of $\pnK^t$ and $\ell\lpnK^t$ without the mild gap. For simplicity, we assume $\mathcal{O} = \emptyset$; it is straightforward to verify that the proof also holds for any oracle $\mathcal{O}$.

    Let $A$ be a probabilistic polynomial-time algorithm that decides $\mathsf{Mild}\textnormal{-}\mathsf{Gap}_{\rho}\textnormal{-}\MINpKT$ with error at most $1/3$, for some polynomial $\rho$. Let $A^m$ denote the algorithm that runs $A$ $m$ times and returns the majority vote. By the Chernoff bound, the error of $A^m$ is bounded by $2^{-m/18}$. 

    Suppose we are given an instance 
    \[
        z \vcentcolon= (x \in \bool^n, 1^s, 1^t, 1^\ell, 1^a, 1^b, 1^c, 1^d)
    \]
    of $\mathsf{Mild}\textnormal{-}\mathsf{Gap}_{\rho}\textnormal{-}\MINlpKT$. 

    On the one hand, if $z$ is a yes-instance, then $\ell\lpK^t_{b/a,d/c}(x) \le s$, and by the definition of $\ell\lpK$, we have
    \[
        \Prob_{\hat r \sim \mathcal{U}_\ell}\mleft[\pK^t_{b/a}(x \circ \hat r) \le s + \ell \mright] \ge \frac{d}{c}.        
    \]
    Then by the correctness of $A$, we get
    \[
        \Prob_{\hat r \sim \mathcal{U}_\ell} \mleft[
            \Prob_A \mleft[
                A^{\lceil 54 \log c \rceil}(x \circ \hat r, 1^{s+\ell}, 1^t, 1^a, 1^b) = 1
            \mright] \geq 1 - \frac{1}{c^3}
        \mright] \ge \frac{d}{c},
    \]
    which, by a union bound, yields
    \begin{equation}\label{eq:lpk-yes-A}
        \Prob_{\hat r \sim \mathcal{U}_\ell, A} \mleft[
            A^{\lceil 54 \log c \rceil}(x \circ \hat r, 1^{s+\ell}, 1^t, 1^a, 1^b) = 1
        \mright] \ge \frac{d}{c} - \frac{1}{c^3}.
    \end{equation}

On the other hand, if $z$ is a no-instance, then 
\[
    \ell\lpK^{t + \rho(n + \ell + a + b + c + d)}_{(b-1)/a, (d-1)/c}(x) \geq s + \log \rho\bigl(t + \ell + a + b + c + d\bigr),
\]
and hence
\[
    \Prob_{\hat r \sim \mathcal{U}_\ell} \mleft[
        \pK^{t + \rho(n + \ell + a + b + c + d)}_{(b-1)/a}(x \circ \hat r) < s + \ell + \log \rho\bigl(t + \ell + a + b + c + d\bigr)
    \mright] < \frac{d}{c} - \frac{1}{c}.
\]
Note that the above implies
\[
    \Prob_{\hat r \sim \mathcal{U}_\ell} \mleft[
        \pK^{t + \rho(n + \ell + a + b)}_{(b-1)/a}(x \circ \hat r) < s + \ell + \log \rho(t + \ell + a + b)
    \mright] < \frac{d}{c} - \frac{1}{c}.
\]
Again, by the correctness of $A$, the above implies
\[
    \Prob_{\hat r \sim \mathcal{U}_\ell} \mleft[
        \Prob_A \mleft[
            A^{\lceil 54 \log c \rceil}(x \circ \hat r, 1^{s+\ell}, 1^t, 1^a, 1^b) = 1
        \mright] > \frac{1}{c^3}
    \mright] < \frac{d}{c} - \frac{1}{c},
\]
which, by an averaging argument, gives
\begin{equation}\label{eq:lpk-no-A}
    \Prob_{\hat r \sim \mathcal{U}_\ell, A} \mleft[
        A^{\lceil 54 \log c \rceil}(x \circ \hat r, 1^{s+\ell}, 1^t, 1^a, 1^b) = 1
    \mright] < \frac{d}{c} - \frac{1}{c} + \frac{1}{c^3}.
\end{equation}

Comparing \Cref{eq:lpk-yes-A} and \Cref{eq:lpk-no-A}, we obtain a polynomial-time algorithm that distinguishes yes- and no-instances with a gap of at least $1/c - 2/{c^3} \geq 1/(2c)$.
Such a gap can be amplified by repeating the algorithm $\poly(c)$ times.
\end{proof}

Next, assume that $\mathsf{Mild}\textnormal{-} \mathsf{Gap}\textnormal{-}\MINlpKT^{\mathsf{SAT}} \in \prBPP$. We will show a result that informally states the following: for every $x, y \in \bool^n$ and for appropriate settings of parameters $\ell$ and $t$, we have
\[
    \pK^{\poly(t)}(y \mid x) \lesssim \pK^{\poly(n), \SAT}(y \mid x) + \ell \lpK^{t, \SAT}(x) - (\ell + \poly(n)) \lpK^{t + \poly(n), \SAT}(x).
\]
Observe the similarity between the above and \Cref{eq:intro_techniques_ideal}. The precise technical lemma is as follows:

\begin{lemma}\label{l:lpnK-cond-pK-bound-cond-pnK}
    Suppose $\mathsf{Mild}\textnormal{-}\mathsf{Gap}\textnormal{-}\MINlpKT^{\SAT} \in \prBPP$. Then there exist a polynomial $p$ and an integer $N_0 \in \mathbb{N}^+$ such that the following holds. For any $n, m, \tau, \ell, t, b, d \in \mathbb{N}^+$ and any strings $x \in \{0,1\}^n$ and $y \in \{0,1\}^m$ satisfying the following constraints\emph{:}
    \begin{itemize}
        \item $2/3 < b/n^3, d/n^3 < 1$,
        \item $N_0 \le n, m\le \tau$,
        \item $\tau \le \ell \le \tau^4$,
        \item $t \ge 2p(\tau)$,
    \end{itemize}
    we have
    \[
        \pK^{p(t)}(y \mid x) \le \pK^{\tau, \SAT}_{1 - 1/n^3}(y \mid x) + \ell \lpK_{b/n^3, (d+1)/n^3}^{t - p(\tau), \SAT}(x) - (\ell + \tau^3) \lpK_{(b-1)/n^3, (d-2)/n^3}^{t + p(\tau), \SAT}(x) + \log p(t).
    \]
\end{lemma}

\begin{proof}[Proof Sketch]
The proof closely follows that of \cite[Lemma 74]{DBLP:journals/eccc/HuL025}, which establishes a similar result for the measure $\ell\lpnK$. The main difference is that the former assumes the easiness of the non-gap version, whereas here we assume the easiness of the mild-gap version. We provide a proof sketch below.

Let $x \in \bool^n $ and $y \in \bool^m$. Similar to the earlier discussion, the idea is to show both upper and lower bounds for the quantity
\[
    \ell\lpK_{b/n^3,\, d/n^3}^{t, \SAT}(x, y).
\]

\paragraph{Upper Bound.}  
\begin{claim}\label{c:lb-lK}
There exists a polynomial $ p_1 $ such that
\[
    \ell\lpK_{b/n^3,\, d/n^3}^{t, \SAT}(x, y) \leq \pK_{1 - 1/n^3}^{\tau, \SAT}(y \mid x) + \ell\lpK_{b/n^3,\, (d+1)/n^3}^{t - p_1(\tau),\SAT}(x) + \log p_1(\tau).
\]
\end{claim}

\begin{proof}[Proof Sketch of \Cref{c:lb-lK}]\renewcommand{\qedsymbol}{$\diamond$}
Intuitively, the claim states that if there is a program that outputs $x $ in time $ t - \poly(\tau)$, and a program that outputs $y $ given $x $ in time $ \tau$, then one can obtain a program that outputs $(x,y)$ in time $ t$. This intuition was shown to hold for $\ell\lpK$ in \cite{DBLP:journals/eccc/HuL025}. See the proof of \cite[Lemma 72]{DBLP:journals/eccc/HuL025} for details.
\end{proof}

\paragraph{Lower Bound.}  
\begin{claim}\label{c:ub-lK} 
There exists a polynomial $p_2$ such that
\[
    \ell\lpK_{b/n^3,\, d/n^3}^{t, \SAT}(x, y) \geq (\ell +\tau^3)\lpK_{(b-1)/n^3,\, (d-2)/n^3}^{t + p_2(\tau), \SAT}(x) + \pK^{p_2(t)}(y \mid x) - \log p_2(t).
\]    
\end{claim}
\begin{proof}[Proof of \Cref{c:ub-lK}]\renewcommand{\qedsymbol}{$\diamond$}
Let $\mathsf{DP}_k \colon \bool^m \times \bool^{mk} \to \bool^{mk + k}$ be the direct product generator (\Cref{def:DPG}), where
\[
    k \vcentcolon= \pK^{q(t)}(y \mid x) - \log q(t),
\]  
and $q$ is a sufficiently large polynomial specified later.\footnote{We may assume without loss of generality that $k>0$. Indeed, if $\pK^{q(t)}(y \mid x) \leq O(\log t)$, then the claim follows trivially from standard monotonicity properties of $\ell\lpK^t$ (see \cite[Lemmas 66 and 67]{DBLP:journals/eccc/HuL025}).}

As in the previous case, we consider two distributions:
\begin{align*}
    \mathcal{D}_1 &\vcentcolon= \bigl(x, \mathsf{DP}_k(y, \mathcal{U}_{mk})\bigr), \\
    \mathcal{D}_2 &\vcentcolon= \bigl(x, \mathcal{U}_{mk + k}\bigr).
\end{align*}
We aim to show that $\mathcal{D}_1$ and $\mathcal{D}_2$ are indistinguishable, and that with high probability over an element from $\mathcal{D}_2$, its complexity is large. To show the latter, we again use weak symmetry of information. However, the key difference here is that we will be able to use a \emph{highly efficient} weak symmetry of information for the notion of $\ell\lpK$. More specifically, we have the following:
For any $ n, m', k, t', \ell \in \mathbb{N}^+$, any $\lambda, \gamma, \alpha \in (0,1)$ satisfying $\alpha < \gamma$, and any string $x \in \{0,1\}^n$, we have
\begin{equation}\label{eq:weak-soi-for-lpk}
    \Prob_{z \sim \mathcal{U}_{m'}} \left[ \ell\lpK_{\lambda,\gamma}^{t', \SAT}(x \circ z) \geq (\ell + m')\lpK_{\lambda, \gamma - \alpha}^{t', \SAT}(x) + m' \right] \geq \alpha.    
\end{equation}
The proof of this result appears in \cite[Lemma 71]{DBLP:journals/eccc/HuL025}.

Now, suppose there exists a (randomized) polynomial-time algorithm $A$ that solves $\mathsf{Mild}\textnormal{-}\mathsf{Gap}_{\rho}\textnormal{-}\MINlpKT^{\mathsf{SAT}}$ for some polynomial $\rho$. Without loss of generality, by applying standard error reduction, we may assume that $A$ has error probability at most $\delta \vcentcolon= 1/(100n^4)$
on every promise instance. Let
\[
    t' \vcentcolon= t + q(\tau) + \rho(\tau^5).
\]
and
\[
    s \vcentcolon= \mleft( (\ell + mk + k) \lpK_{(b-1)/n^3,(d-2)/n^3}^{t', \SAT}(x) + mk + k \mright) - \log \rho(t^4).
\]
We get
\begin{align*}
    &\quad \Prob_{w \sim \mathcal{U}_{mk+k}} \mleft[
        \ell\lpK_{(b-1)/n^3,(d-1)/n^3}^{t + q(\tau) + \rho(n+mk+k+\ell + b+ d+ 2n^3), \SAT}(x \circ w) \ge 
        s + \log\rho(t+q(\tau)+\ell + b+ d+ 2n^3)
    \mright] \\
    &\geq \Prob \mleft[
        \ell\lpK_{(b-1)/n^3,(d-1)/n^3}^{t + q(\tau) + \rho(\tau^5), \SAT}(x \circ w) \ge 
        s + \log \rho(t^4)
    \mright]\\
        &= \Prob \mleft[
        \ell\lpK_{(b-1)/n^3,(d-1)/n^3}^{t', \SAT}(x \circ w) \ge 
        (\ell + mk + k)\lpK_{(b-1)/n^3,(d-2)/n^3}^{t', \SAT}(x) + mk + k
    \mright]\\
    &\ge \frac{1}{n^3}. \tag{by \Cref{eq:weak-soi-for-lpk} with $\alpha=1/n^3$ and $\lambda=(b-1)/n^3$}
\end{align*}
In fact, for the first inequality, we need the notion $\ell \lpK^t$ to be monotone in $t$, meaning that $\ell \lpK^t$ decreases as $t$ increases (see \cite[Lemma 66]{DBLP:journals/eccc/HuL025}).

Note that the above implies that, with probability at least $1/n^3$ over $(x, w)\sim\mathcal{D}_2$, $A$ (with probability at least $1-\delta$ over its internal randomness) outputs $0$ on $(x\circ w, 1^s, 1^{t + q(\tau)}, 1^\ell, 1^{n^3}, 1^b, 1^{n^3}, 1^d)$. This yields
\begin{equation*}
    \Prob_{(x, w)\sim\mathcal{D}_2, A}\mleft[A(x\circ w, 1^s, 1^{t + q(\tau)}, 1^\ell, 1^{n^3}, 1^b, 1^{n^3}, 1^d)=0\mright]\geq (1-\delta)\cdot \frac{1}{n^3}. 
\end{equation*}

For the sake of contradiction, suppose the procedure $A(-, 1^s, 1^{t + q(\tau)}, 1^\ell, 1^{n^3}, 1^b, 1^{n^3}, 1^d)$
can $(1/n^4)$-distinguish $\mathcal{D}_1$ and $\mathcal{D}_2$, i.e,
\begin{equation}\label{eq:distin}
    \Prob_{z\sim\mathcal{U}_{mk}, A}\mleft[A(x\circ \mathsf{DP}_k(y, z), 1^s, 1^{t + q(\tau)}, 1^\ell, 1^{n^3}, 1^b, 1^{n^3}, 1^d)=0\mright]< (1-\delta)\cdot \frac{1}{n^3} - \frac{1}{n^4}.
\end{equation}
Then, by the fact that $A$ is a polynomial-time algorithm and by the reconstruction property of $\mathsf{DP}_k$ (\Cref{lem:recon}; see also \cite[Corollary 24]{DBLP:journals/eccc/HuL025}), we would obtain
\[
    \pK^{\poly(t)}(y \mid x) \le k + O(\log t),
\]
which contradicts our choice of $k$, provided that $q$ is chosen to be a sufficiently large polynomial.

In other words, the quantity in \Cref{eq:distin} is at least $(1-\delta)\cdot n^{-3} - n^{-4} > \delta$.
It follows that there exists $z \in \bool^{mk}$ such that
\begin{equation}\label{eq:s-bound}
    \ell\lpK^{t + q(\tau), \SAT}_{b/n^3, d/n^3} \mleft( x, \mathsf{DP}_k(y, z) \mright) \ge s.
\end{equation}
Moreover, it can also be shown (see \cite[Lemma 70]{DBLP:journals/eccc/HuL025}) that, for $q$ a sufficiently large polynomial,
\begin{equation}\label{eq:DK-bound}
    \ell\lpK^{t + q(\tau), \SAT}_{b/n^3, d/n^3} \mleft( x, \mathsf{DP}_k(y, z) \mright) \le \ell\lpK_{b/n^3, d/n^3}^{t, \SAT}(x, y) + mk + \log q(\tau).
\end{equation}
Intuitively, this is because $\mathsf{DP}_k(y, z)$ is efficiently computable given $y$ and the seed $z$.

It follows from \Cref{eq:s-bound} and \Cref{eq:DK-bound} that
\begin{align*}
    \ell \lpK_{b/n^3,\, d/n^3}^{t,\SAT}(x, y) &\geq s - mk - \log q(\tau)\\
    &= (\ell + mk + k) \lpK_{(b-1)/n^3,(d-2)/n^3}^{t', \SAT}(x) + k - \log \rho(t^4) - \log q(\tau) \\
    &= (\ell + mk + k) \lpK_{(b-1)/n^3,\, (d-2)/n^3}^{t+q(\tau)+\rho(\tau^5), \SAT}(x) + \pK^{q(t)}(y \mid x) \\
    &\qquad - \log q(t) -  \log \rho(t^4) - \log q(\tau)\\
    &\geq (\ell + \tau^3) \lpK_{(b-1)/n^3,\, (d-2)/n^3}^{t+q(\tau^5)+\rho(\tau^5), \SAT}(x) + \pK^{q(t)}(y \mid x) \\
    &\qquad - \log q(t) -  \log \rho(t^4) - \log q(\tau^5),
\end{align*}
where the last inequality uses the monotonicity of $\ell \lpK^t$ in $\ell$ (see \cite[Lemma 67]{DBLP:journals/eccc/HuL025}), which roughly states that for any sufficiently large polynomial $q$, we have that for every $\ell<\ell'$, $x\in\{0,1\}^*$ and $t'\in\mathbb{N}$,
    \[
        \ell'\lpK^{t'+q(|x|+|y|+\ell'),\SAT}(x)\le\ell\lpK^{t',\SAT}(x)+\log q(|x|+|y|+\ell').
    \]
Finally, we can conclude that there exists a polynomial $p_2$ such that
\[
    \ell \lpK_{b/n^3,\, d/n^3}^{t, \SAT}(x, y) \geq (\ell + \tau^3) \lpK_{(b-1)/n^3,\, (d-2)/n^3}^{t + p_2(\tau), \SAT}(x) + \pK^{p_2(t)}(y \mid x) - \log p_2(t).
\]
This completes the proof of the lower bound.
\end{proof}

\paragraph*{Putting It All Together.}  
By combining the upper and lower bounds (\Cref{c:lb-lK} and \Cref{c:ub-lK}) and rearranging, we obtain
\begin{align*}
    \pK^{p_2(t)}(y \mid x) &\leq \pK_{1 - 1/n^3}^{\tau, \SAT}(y \mid x) + \ell \lpK_{b/n^3,\, (d+1)/n^3}^{t - p_1(\tau), \SAT}(x) - (\ell + \tau^3) \lpK_{(b-1)/n^3,\, (d-2)/n^3}^{t + p_2(\tau), \SAT}(x) \\
    &\qquad + \log p_1(\tau) + \log p_2(t),
\end{align*}
which, by the monotonicity of $\ell \lpK^t$ with respect to $t$, implies the existence of a polynomial $p$ such that
\begin{align*}
    \pK^{p(t)}(y \mid x) \leq \pK_{1 - 1/n^3}^{\tau, \SAT}(y \mid x) + \ell \lpK_{b/n^3,\, (d+1)/n^3}^{t - p(\tau),\SAT}(x) - (\ell + \tau^3) \lpK_{(b-1)/n^3,\, (d-2)/n^3}^{t + p(\tau), \SAT}(x) + \log p(t),
\end{align*}
as desired.
\end{proof}

\begin{proof}[Proof Sketch of \Cref{t:MINpnKT-NP-hard}] 
The proof is essentially the same as that of \cite[Lemma 75]{DBLP:journals/eccc/HuL025}. We sketch the high-level ideas here and refer the reader to \cite{DBLP:journals/eccc/HuL025} for full details.

Suppose $\mathsf{Mild}\textnormal{-}\mathsf{Gap}\textnormal{-}\MINpKT^{\SAT} \in \prBPP$, which, by \Cref{l:MINlpnKT-from-MINpnKT}, implies $\mathsf{Mild}\textnormal{-}\mathsf{Gap}\textnormal{-}\MINlpKT^{\SAT} \in \prBPP$. Given \Cref{l:lpnK-cond-pK-bound-cond-pnK}, we can apply a slightly more sophisticated telescoping sum argument than described earlier and show that for every $x \in \bool^n $ and any $\tau$ that is a large enough polynomial in $ n$, there exist $t, \ell, b, d$ that are at most $\poly(n)$ such that  
\[
    \ell \lpK_{b/n^3, (d+1)/n^3}^{t - p(\tau), \SAT}(x) - (\ell + \tau^3) \lpK_{(b-1)/n^3, (d-2)/n^3}^{t + p(\tau), \SAT}(x) + \log p(t)
\]  
is at most $O(\log n)$. Again, this can be achieved because the increase in the time bounds between the two quantities is only an \emph{additive} $\poly(n)$ term. As a result, for every $x$ and $y$, we obtain
\begin{equation}\label{eq:pK-pnK-2}
    \pK^{\poly(n)}(y \mid x) \leq \pK^{\tau, \SAT}_{1 - 1/n^3}(y \mid x) + O(\log n).
\end{equation}

Also, as mentioned earlier, $\pK^{t,\SAT}$ admits the language compression property (see \cite[Theorem 4]{DBLP:journals/eccc/HuL025}), that is, for every set $A_x$ whose membership can be decided in polynomial time given $x$, there exists polynomial $\tau$ such that
\begin{equation}\label{eq:pnK-LC}
    \pK^{\tau, \SAT}_{1 - 1/n^3}(y \mid x) \leq \log |A_x| + O(\log n).
\end{equation}

Now let $L \in \NP$ and consider any $x \in \bool^n$ that is a yes-instance of $L$. Let $A_x$ be the set of $L$-witnesses for $x$. Then \Cref{eq:pK-pnK-2,eq:pnK-LC} imply that for every $y \in A_x$,
\[
    \pK^{p(n)}(y \mid x) \leq \log |A_x| + O(\log n).
\]
As described earlier, this enables an efficient sampling procedure which, given $x$, outputs some element in $A_x$ with non-trivial probability $ 1/\poly(n)$. Repeating this procedure sufficiently many times then yields a witness for $x $ with high probability.
\end{proof}

\subsection{Worst-Case Easiness of \texorpdfstring{$\mathsf{Gap}\textnormal{-}\mathsf{MINpKT}^{\mathsf{PH}}$}{Gap-MINpKT-PH} from Average-Case Easiness of \texorpdfstring{$\mathsf{PH}$}{PH}}\label{sec:PH-to-GapMINpnKT}

\begin{lemma}\label{l:PH-to-GapMINpnKT}
    If $\Dist\mathsf{PH}\subseteq\Avg\BPP$, then $\mathsf{Gap}\textnormal{-}\mathsf{MINpKT}^{\mathsf{PH}}\in\prBPP$.
\end{lemma}

To show \Cref{l:PH-to-GapMINpnKT}, we need the following.
\begin{lemma} \label{worst-to-avg-Kt-PH}
    For every oracle $\mathcal{O}$, if $(\mathsf{MINKT^{\mathcal{O}}},\mathsf{PSAMP}) \subseteq \mathsf{AvgBPP}$, then $\mathsf{Gap}\textnormal{-}\mathsf{MINpKT}^{\mathcal{O}}\in\prBPP$.    
\end{lemma}
\begin{proof}[Proof Sketch]
The proof is similar to that of {\cite[Proposition 58]{DBLP:journals/eccc/HuL025}}, using techniques from worst-case-to-average-case reductions developed by \cite{DBLP:conf/focs/Hirahara18}.
\end{proof}

\begin{proof}[Proof of \Cref{l:PH-to-GapMINpnKT}]
    Suppose $\Dist\mathsf{PH}\subseteq\Avg\BPP$. Then for every oracle $\mathcal{O}\in\mathsf{PH}$, $(\mathsf{MINKT^{\mathcal{O}}},\mathsf{PSAMP}) \in  \mathsf{AvgBPP}$, since $\mathsf{MINKT}^{\mathcal{O}}\in\mathsf{PH}$. This, by \Cref{worst-to-avg-Kt-PH}, implies that $\mathsf{Gap}\textnormal{-}\mathsf{MINpKT}^{\mathcal{O}}\in\prBPP$. It follows that $\mathsf{Gap}\textnormal{-}\mathsf{MINpKT}^{\mathsf{PH}}\in\prBPP$.
\end{proof}

\subsection{Average-Case Easiness of \texorpdfstring{$\mathsf{PH}$}{PH} from Worst-Case Easiness of \texorpdfstring{$\mathsf{Gap}\textnormal{-}\mathsf{MINpKT}^{\mathsf{PH}}$}{Gap-MINpKT-PH}}\label{sec:GapMINpnKT-to-PH}

In this subsection, we prove the following which shows the forward direction of \Cref{t:PH-heuristica-characterization}. 
 
\begin{lemma}\label{l:GapMINpnKT-to-PH}
        If $\mathsf{Gap}\textnormal{-}\mathsf{MINpKT}^{\mathsf{PH}}\in\prBPP$, then $\Dist\mathsf{PH}\subseteq\Avg\BPP$.
\end{lemma}

We begin with some technical tools.

\subsubsection{Technical Tools}

The following is the main technical lemma for this subsection.

\begin{lemma}\label{l:alg-Sigma-k-cd}
    Let $k_0\in\mathbb{N}$, and let $\mathcal{O}$ be a $\Sigma^{\P}_{k_0}$-complete language.
    If there exist polynomials $\rho$ and $\tau$ such that $\mathsf{Gap}_\rho$-$\mathsf{MINpKT}^{\mathcal{O}}\in\mathsf{prBPTIME}[\tau(n)]$, then for every $k\leq k_0$, $L\in\Sigma^{\P}_{k}$, and polynomial $q$, there exists a randomized algorithm $A$ such that for every $\ell\in \mathbb{N}$ and $x\in\bool^n$, with probability at least $1-2^{-\ell}$ \emph{(}over the internal randomness of $A$\emph{)} both of the following hold. 
    \begin{enumerate}
        \item $A(x,1^{\ell})$ runs in time \[
        2^{O\mleft(\pK^{q(n), \mathcal{O}}(x) - \K(x)\mright)}\cdot \poly(n,\ell),
        \]
        where the constants hidden in $O(\cdot)$ and $\poly(\cdot)$ depend only on $\mathcal{O}$, $\rho$, $\tau$, $k$, $L$, and $q$.
        \item  $A(x,1^{\ell})=1$ if and only if $x\in L$.
    \end{enumerate}
\end{lemma}

The rest of this subsection is devoted to proving \Cref{l:alg-Sigma-k-cd}. For this, we need the following additional tools.

	\begin{lemma}[Symmetry of Information for $\pK^t$; see, e.g., {\cite[Lemma 36]{DBLP:conf/coco/HiraharaKLO24}}]\label{l:soi-pKt}
	Let $\mathcal{O}$ be any oracle. If     $\mathsf{Gap}\textnormal{-}\mathsf{MINpKT}^{\mathcal{O}}\in\prBPP$, then there exist polynomials $p_{\sf _{SoI}}$ and $p_0$ such that for all sufficiently large $x, y \in \{0,1\}^*$ and every $t \geq p_0(|x|+|y|)$, 
	\[
	\pK^{p_{\sf _{SoI}}(t)}(y \mid x) \leq \pK^{t}(x,y)  - \pK^{p_{\sf _{SoI}}(t)}(x) + \log p_{\sf _{SoI}}(t).
	\]
\end{lemma}
\begin{proof}[Proof Sketch]
    To get the desired symmetry of information property, it suffices have an efficient algorithm that accepts the set of strings with low $\mathsf{K}^{\poly}$-complexity while rejecting a large fraction of random strings. Note that such an algorithm follows from $\mathsf{Gap}\textnormal{-}\mathsf{MINpKT}^{\mathcal{O}}\in\prBPP$ for any $\mathcal{O}$.
\end{proof}

	\begin{lemma}\label{l:pK-pnK-3}
		Let $\mathcal{O}$ be any oracle. If     $\mathsf{Gap}\textnormal{-}\mathsf{MINpKT}^{\mathcal{O}}\in\prBPP$, then there exists a polynomial $p_1$ such that for all $z\in\bool^*$ and $t\in\mathbb{N}$, 
		\[
			\pK^{p_1(t)}(z)\leq \pK^{t, \mathcal{O}}(z) + \log p_1(t).
		\]
	\end{lemma}
	\begin{proof}[Proof Sketch]
		The proof can be easily adapted from that of {\cite[Lemma 56]{DBLP:journals/eccc/HuL025}}.
	\end{proof}

    We are now ready to prove \Cref{l:alg-Sigma-k-cd}.

\begin{proof}[Proof of \Cref{l:alg-Sigma-k-cd}]
    Let $k_0 \in \mathbb{N}$ and let $\mathcal{O}$ be a $\Sigma^{\P}_{k_0}$-complete language. Suppose $\mathsf{Gap}_\rho$-$\mathsf{MINpKT}^{\mathcal{O}}\in\mathsf{prBPTIME}[\tau(n)]$ for some polynomials $\rho$ and $\tau$. The proof is by induction on $k$.

The base case for $k = 0$ is trivial.

For the induction step, suppose the claimed conclusion holds for $k - 1$, where $1 \leq k \le k_0$. Let $L\in \Sigma^{\P}_{k}$ and let $V$ be a corresponding verifier. That is, for $x \in \{0, 1\}^n$, $x \in L$ if and only if
\[
\exists y_1 \in \bool^{v(n)}, \forall y_2 \in \bool^{v(n)}, \dots, Q^{k} y_k \in \bool^{v(n)}, V(x, y_1, y_2, \dots, y_k) = 1,
\]
where $v$ is a polynomial and $Q^{k}\in\{\exists,\forall\}$. Let $q$ be a polynomial.

To decide $L$, we aim to find a witness for a given instance, if one exists. Consider a yes-instance $x \in \{0,1\}^n$ of $L$. Let $y_1^*$ be the lexicographically first witness for $x$, i.e., $y_1^*$ is the lexicographically smallest string in $\bool^{v(n)}$ that satisfies 
\[
\forall y_2 \in \bool^{v(n)}, \dots, Q^{k} y_k \in \bool^{v(n)}, V(x, y_1, y_2, \dots, y_k) = 1,
\]
where $Q^k \in \{\exists, \forall\}$.
Let $t$ be a sufficiently large polynomial specified later. We first observe the following.
\begin{claim}\label{c:fist-witness-small-cd}
    There exists a polynomial $p$ such that 
    \[
        \pK^{p(t(n))}(y^*_1 \mid x) \leq \pK^{t(n),\mathcal{O}}(x) - \pK^{p(t(n)),\mathcal{O}}(x) + \log p(t(n)).
    \]
\end{claim}
\begin{proof}[Proof of \Cref{c:fist-witness-small-cd}]\renewcommand{\qedsymbol}{$\diamond$}
Note that by \Cref{l:soi-pKt} and the assumption that $\mathsf{Gap}$-$\mathsf{MINpKT}^{\mathcal{O}}$ is easy, we get symmetry of information for $\pK^{t}$. Now let $p_1$ be the polynomial from \Cref{l:pK-pnK-3} and let $t\vcentcolon=t(n)$. We have
    \begin{align*}
    	    &\quad\pK^{p_{\sf _{SoI}}(p_1(t+\poly(n)))}(y^*_1 \mid x) \\
            &\leq \pK^{p_1(t+\poly(n))}(x,y^*_1)  - \pK^{p_{\sf _{SoI}}(p_1(t+\poly(n)))}(x) + \log p_{\sf _{SoI}}(p_1(t+\poly(n))) \\
    	    &\leq \pK^{t+\poly(n),\mathcal{O}}(x,y^*_1)  - \pK^{p_{\sf _{SoI}}(p_1(t+\poly(n)))}(x) + \log p_{\sf _{SoI}}^2(p_1(t+\poly(n))) \tag{by \Cref{l:pK-pnK-3}}\\
    	    &\leq \pK^{\poly(n), \mathcal{O}}(y^*_1\mid x) +  \pK^{t,\mathcal{O}}(x) - \pK^{p_{\sf _{SoI}}(p_1(t+\poly(n)))}(x) + \log p_{\sf _{SoI}}^3(p_1(t+\poly(n)))\\
            &\leq \pK^{\poly(n), \mathcal{O}}(y^*_1\mid x) +  \pK^{t,\mathcal{O}}(x) - \pK^{p_{\sf _{SoI}}(p_1(t+\poly(n))),\mathcal{O}}(x) + \log p_{\sf _{SoI}}^3(p_1(t+\poly(n))).
    \end{align*}
	Let $p_2$ be a sufficiently large polynomial, the above yields
	\[
		\pK^{p_2(t)}(y^*_1\mid x)\leq \pK^{\poly(n), \mathcal{O}}(y^*_1\mid x) + \pK^{t,\mathcal{O}}(x) -\pK^{p_2(t),\mathcal{O}}(x) + \log p_2(t),
	\]
Now note that $\pK^{\poly(n),\mathcal{O}}(y^*_1 \mid x)\leq O(1)$ since given $x$ and oracle access to $\mathcal{O}$, one can recover the lexicographically first witness for $x$ in time $\poly(n)$. Therefore, we get
\begin{align*}
    \pK^{p_2(t)}(y^*_1 \mid x) \leq O(1) + \pK^{t,\mathcal{O}}(x) - \pK^{p_2(t),\mathcal{O}}(x) + \log p_2(t).
\end{align*}    
The claims follows by choosing $p$ to be a sufficiently large polynomial.
\end{proof}

We will generate a set of candidate strings that is guaranteed to include a witness for $x$ (specifically, $y_1^*$) if $x \in L$. Given this, it suffices to solve the task of checking whether a given string is a witness for $x$. This is achieved using the language $L' \in \Sigma^{\mathrm{P}}_{k-1}$ defined as follows.

\begin{quote}
    $(x,y_1) \in L'$ if and only if $\exists y_{2}, \dots, Q^{k} y_{k}, V(x,y_1,y_2,\dots, y_{k})=0$, where $Q^k \in \{\exists, \forall\}$.
\end{quote}

Let $q'> q$ be a sufficiently large polynomial specified later. We have the following.
\begin{claim}\label{alg-A-prime}
There exists a randomized algorithm $A'$ such that every $\ell'\in \mathbb{N}$, $x\in\bool^n$ and $y_1\in\bool^{v(n)}$, with probability at least $1-2^{-\ell'}$ \emph{(}over the internal randomness of $A'$\emph{)} both of the following hold.
\begin{enumerate}
    \item $A'\mleft((x,y_1),1^{\ell'}\mright)$ runs in time 
        \[
        2^{O\mleft(\pK^{q'(n+v(n)), \mathcal{O}}(x,y_1) - \K(x,y_1)\mright)}\cdot \poly(\ell',n).
        \]
    \item $A'\mleft((x,y_1),1^{\ell'}\mright)=1$ if and only if  $(x,y_1)\in L'$.
\end{enumerate}
\end{claim}
\begin{proof}[Proof of \Cref{alg-A-prime}]\renewcommand{\qedsymbol}{$\diamond$}
    This follows directly from our induction hypothesis.
\end{proof}

     In order to generate a set of ``good'' candidate witnesses, we will also need the following subroutines.

    \begin{claim}\label{c:alg-depth}
There exists a randomized polynomial-time algorithm $\mathsf{Approx}\text{-}\mathsf{depth}$ such that, for every $x\in \bool^n$ and  $\ell' \in\mathbb{N}$, $\mathsf{Approx}\text{-}\mathsf{depth}(x,\ell')$ outputs, with probability at least $1 - 2^{-\ell'}$, an integer $s$ satisfying
\[
    \pK^{t(n), \mathcal{O}}(x) - \pK^{p(t(n)), \mathcal{O}}(x) \leq s \leq \pK^{\rho^{-1}(t(n)), \mathcal{O}}(x) - \pK^{\rho(p(t(n))), \mathcal{O}}(x) + \log t(n) + \log \rho(p(t(n))),
\]
where $p$ is the polynomial defined in \Cref{c:fist-witness-small-cd}.
    \end{claim}
    \begin{proof}[Proof of \Cref{c:alg-depth}]\renewcommand{\qedsymbol}{$\diamond$}
Let $\mathsf{Approx}$ be a randomized polynomial-time algorithm that satisfies the following: Given $z\in\bool^m$, $1^{t'}$ and $1^{\kappa}$, $\mathsf{Approx}(z,1^{t'},1^{\kappa})$ outputs, with probability at least $1-2^{-\kappa}$, an integer $s_0$ such that
\[
      \pK^{\rho(t'), \mathcal{O}}(z)-\log \rho(t') \leq s_0\leq \pK^{t', \mathcal{O}}(z).
\]
It is easy to see that such an algorithm $\mathsf{Approx}$ can be obtained from the assumption that $\mathsf{Gap}_{\rho}\textnormal{-}\mathsf{MINpKT}^{\mathcal{O}}\in\prBPP$ and using success amplification techniques. (See, e.g., {\cite[Lemma 28]{GKLO22}}.) 

Let $t\vcentcolon=t(n)$. By running $\mathsf{Approx}(x,1^{\rho^{-1}(t)}, 1^{\ell'+1})$, with probability at least $1-2^{-\ell'}/2$, we get an integer $s_1$ such that
		\[
		\pK^{t,\mathcal{O}}(x)- \log t \;\leq\; s_1 \;\leq\; \pK^{\rho^{-1}(t),\mathcal{O}}(x).
		\]
		Similarly, by running $\mathsf{Approx}(x,1^{p(t)}, 1^{\ell+1})$, with probability at least $1-2^{-\ell'}/2$, we get some integer $s_2$ such that
		\[
		\pK^{\rho(p(t)),\mathcal{O}}(x)- \log\rho(p(t))\;\leq\; s_2 \;\leq\; \pK^{p(t),\mathcal{O}}(x).
		\]
		Then by a union bound, with probability at least $1-2^{-\ell'}$, we have
		\begin{equation}\label{eq:s1-s2-up}
			s_1-s_2\leq \pK^{\rho^{-1}(t),\mathcal{O}}(x) - \mleft(\pK^{\rho(p(t)),\mathcal{O}}(x)- \log\rho(p(t)) \mright)            
		\end{equation}
		and
		\begin{equation}\label{eq:s1-s2-low}
			s_1-s_2\geq \mleft(\pK^{t,\mathcal{O}}(x)- \log(t)\mright) - \pK^{p(t),\mathcal{O}}(x).
		\end{equation}
		We can then output 
		\[
		s \vcentcolon=s_1-s_2+\log(t).
		\]
        It is easy to verify that $s$ satisfies the stated condition in the claim.
    \end{proof}

    \begin{claim}\label{c:alg-valid}
There exist a randomized polynomial-time algorithm $\mathsf{Valid}$ and a constant $C$ such that for every $(x, y_1, 1^{\ell'}) \in \bool^{n} \times \bool^{v(n)} \times \mathbb{N}$, with probability at least $1 - 2^{-\ell'}$ \emph{(}over the internal randomness of $\mathsf{Valid}$\emph{)} both of the following hold\emph{:}
\begin{enumerate}
    \item\label{i:lex-is-valid} If $y_1 = y_1^{*}$, then $\mathsf{Valid}(x, y_1, 1^{\ell'}) = 1$.
    \item\label{i:valid-is-good} If $\mathsf{Valid}(x, y_1, 1^{\ell'}) = 1$, then it holds that
    \[
        \pK^{q'(n+v(n)), \mathcal{O}}(x, y_1) \leq \pK^{q(n), \mathcal{O}}(x) + \log q'(n+v(n)) + \log \rho(q(n)) +C\cdot\log n.
    \]
\end{enumerate}
    \end{claim}
    \begin{proof}[Proof of \Cref{c:alg-valid}]\renewcommand{\qedsymbol}{$\diamond$}
Again, let $\mathsf{Approx}$ be a randomized polynomial-time algorithm that satisfies the following: Given $z\in\bool^m$, $1^{t'}$ and $1^{\kappa}$, $\mathsf{Approx}(z,1^{t'},1^{\kappa})$ outputs an integer $s$ such that
\[
      \pK^{\rho(t'), \mathcal{O}}(z)-\log \rho(t') \leq s\leq \pK^{t', \mathcal{O}}(z),
\]
with probability at least $1-2^{-\kappa}$, where $\rho$ is the polynomial such that $\mathsf{Gap}_{\rho}\textnormal{-}\mathsf{MINpKT}^{\mathcal{O}}\in\prBPP$.
    
    Let $\mathsf{Valid}$ be the following algorithm, where $C$ is a constant to be determined later.
    \begin{quote}
        Given $(x,y_1,1^{\ell'})$, let \[
        \theta\vcentcolon=\mathsf{Approx}\mleft(x,1^{q(n)}, 1^{\ell'+1}\mright),
        \]
        and 
        \[
        \mu\vcentcolon=\mathsf{Approx}\mleft((x, y_1), 1^{\rho^{-1}(q'(n+v(n)))}, 1^{\ell'+1}\mright).
        \]
        We accept if and only if $\mu\leq \theta + \log \rho(q(n))+C\cdot\log n$.
    \end{quote}
    We argue that the above algorithm satisfies the two conditions stated in the claim. Indeed, by the correctness of $\mathsf{Approx}$ and a union bound, we have that, with probability at least $1-2^{-\ell'}$, both the following two hold.
    \[
            \pK^{\rho(q(n)), \mathcal{O}}(x)-\log \rho(q(n)) \leq \theta \leq \pK^{q(n), \mathcal{O}}(x),
    \]
    and 
    \[
            \pK^{q'(n+v(n)), \mathcal{O}}(x,y_1)-\log q'(n+v(n)) \leq \mu \leq \pK^{\rho^{-1}(q'(n+v(n))), \mathcal{O}}(x,y_1).
    \]
    
    Now to see that $\mathsf{Valid}(x,y^*_1, 1^{\ell})=1$, observe that in this case 
\begin{align*}
     \mu &\leq \pK^{\rho^{-1}(q'(n+v(n))), \mathcal{O}}(x,y_1^{*}) \\
     &\leq \pK^{\rho(q(n)), \mathcal{O}}(x)+C\cdot\log n\\
     &\leq \theta + \log \rho(q(n))+C\cdot\log n,    
\end{align*}
where the second inequality uses the facts that $q'$ is a sufficiently large polynomial and that given oracle access to $\mathcal{O}$, one can recover $y_1^{*}$ in polynomial time given $x$ with a program of length $C\cdot\log n$, for some constant $C$.

For the second condition, if $\mathsf{Valid}(x,y_1,1^{\ell})=1$, then it implies that 
\begin{align*}
\pK^{q'(n+v(n)), \mathcal{O}}(x,y_1)-\log q'(n+v(n)) &\leq \mu \\
&\leq \theta +  \log \rho(q(n)) +C\cdot\log n \\
&\leq \pK^{q(n), \mathcal{O}}(x) +  \log \rho(q(n)) +C\cdot\log n.
\end{align*}
This completes the proof.
\end{proof}

We are now ready to describe our algorithm for deciding $L$.

    		\begin{algorithm}[H]
			\caption{An algorithm for deciding $L$}
			\begin{algorithmic}[1]
				
				\Procedure{$A$}{$x, 1^{\ell}$} 
				\State $n \vcentcolon= |x|$.
				\State $s\vcentcolon= \mathsf{Approx}\textnormal{-}\mathsf{depth}(x, 1^{\ell+3})$.
                \State $s'\vcentcolon= s+ \log p(t(n))$.
            
                \\
                \For{$i \in [\ell]$}
                                \State $r \vcentcolon= \text{a uniformly random string in } \bool^{p(t(n))}$.
				\For{$\Pi \in \bool^{\leq s'}$}
				\State $y_1 \vcentcolon= \text{the output of } \Pi(x;r)$ after running $p(t(n))$ steps.
				\If{$|y_1|=v(n)$ and $\mathsf{Valid}\mleft(x,y_1,1^{\ell + \log (\ell)+s'+5}\mright)=1$}
				\If {$A'\mleft(x,y_1,1^{\ell + \log (\ell)+s'+5}\mright)=0$}
                 \State Output $\mathsf{True}$.
                \EndIf
                
				\EndIf
				\EndFor
                \EndFor
				\State Output $\mathsf{False}$.
				\EndProcedure
			\end{algorithmic}
		\end{algorithm}

\paragraph{Correctness.}
We argue that the algorithm $A$ correctly decides $L$. 
Fix $x\in\bool^n$ and $\ell\in\mathbb{N}$.
        
First, we assume that the (randomized) algorithm $\mathsf{Approx}$-$\mathsf{depth}$ outputs a correct answer for $x$—that is, it satisfies the condition stated in \Cref{c:alg-depth}. Note that this occurs with probability at least $1 - 2^{-\ell}/8$, due to the choice of the second parameter in the invocation of $\mathsf{Approx}$-$\mathsf{depth}$.
We also assume that both algorithms, $\mathsf{Valid}$ and $A'$, succeed (meaning that the conditions in \Cref{c:alg-valid} and \Cref{alg-A-prime}, respectively, are satisfied) in all of their at most $\ell \cdot 2^{s'+1}$ executions. By a union bound and the parameter choices in their invocations, this occurs with probability at least $1 - 2^{-\ell}/8$.

Now consider a single execution of the for loop at Line~$6$ (i.e., Lines $7$--$12$).
Suppose $x \not\in L$. Note that, in this case, by the definition of the language $L'$ and the fact that $A'$ correctly decides $L'$ (\Cref{alg-A-prime}), the condition at Line~$11$ will not be satisfied for any $y_1$. Therefore, Line~$12$ will never be reached in this case.
        
        On the other hand, if $x\in L$, then by \Cref{c:fist-witness-small-cd}, with probability at least $2/3$ over $r\sim\bool^{p(t(n))}$, the multiset
        \begin{equation}\label{eq:candidate_witness}
        \mathcal{S}\vcentcolon = \mleft\{ y_1 \mid \textnormal{$y_1$ is the output of $\Pi(x;r)$ after $p(t(n))$ steps for some $\Pi\in\bool^{\leq s'} $}\mright\},
        \end{equation}
        generated by Line~$9$, contains $y^*_1$, which is the lexicographically first witness for $x$. To see this, observe that if the algorithm $\mathsf{Approx}$-$\mathsf{depth}$ outputs a correct answer $s$ for $x$, then we have
        \begin{align*}
         \pK^{p(t(n))}(y^*_1 \mid x) &\leq \pK^{t(n),\mathcal{O}}(x) - \pK^{p(t(n)),\mathcal{O}}(x) + \log p(t(n)) \tag{by \Cref{c:fist-witness-small-cd}}\\
         &\leq s + \log p(t(n))\tag{by \Cref{c:alg-depth}}\\
         &\leq s'.
        \end{align*}
Consider the case when $y_1 = y^*_1$. It follows from \Cref{c:alg-valid} (\Cref{i:lex-is-valid}) that the condition at Line~$10$ is satisfied. Also, since $y^*_1$ is a witness for $x$, by the definition of $L'$, we have $(x, y^*_1) \not\in L'$, and hence, by the fact that $A'$ decides $L'$ (\Cref{alg-A-prime}), the condition at Line~$11$ will be satisfied. Therefore, $\mathsf{True}$ will be output in this case.

Since the above process is executed $\ell$ times in the for loop, the probability that it never outputs $\mathsf{True}$ when $x \in L$ is at most $1/3^{\ell}$.

        To conclude, with probability at least 
        \[
            1-\frac{1}{8\cdot 2^{\ell}}-\frac{1}{8\cdot 2^{\ell}} - \frac{1}{3^{\ell}}> 1-\frac{1}{2^{\ell}},
        \]
        the algorithm will correctly decide whether $x\in L$.

        \paragraph{Running Time.} We now analyze the running time of the algorithm. Again, we assume that the (randomized) algorithms $\mathsf{Approx}$-$\mathsf{depth}$, $\mathsf{Valid}$, and $A'$ succeed in all their executions, which occurs with probability at least $1 - 2^{-\ell}$.

Consider a single execution of the for loop at Line $6$. 
Define the following sub-multiset of $\mathcal{S}$ (where $\mathcal{S}$ is defined in \Cref{eq:candidate_witness}):
\[
\mathcal{S'} \vcentcolon = \mleft\{ y_1 \in \mathcal{S} \mid \textnormal{$y_1$ satisfies the condition at Line $10$} \mright\}.
\]
Now, it is easy to see that a single execution (Lines $7$--$12$) takes time at most
\begin{equation}\label{eq:total-runtime}
    2^{s'+1} \cdot \poly(t(n)) + \sum_{y_1 \in \mathcal{S'}} \mathsf{Runtime}\mleft(A'(x,y_1,1^{\ell'})\mright),
\end{equation}
where $\ell'= {\ell + \log (\ell)+s'+5}$.

Note that for every $y_1 \in \mathcal{S'}$, it follows from \Cref{c:alg-valid} (\Cref{i:valid-is-good}) that
\begin{equation}\label{eq:good-witness-runtime}
    \pK^{q'(n+v(n)), \mathcal{O}}(x,y_1) \leq \pK^{q(n), \mathcal{O}}(x) + \log q'(n+v(n)) + \log \rho(q(n)) + C\cdot\log n.
\end{equation}
        Therefore, for every such $y_1$, we have
        \begin{align*}
            \quad \mathsf{Runtime}\mleft(A'(x,y_1,1^{\ell'})\mright) &\leq            \exp\mleft(O\mleft(\pK^{q'(n+v(n)), \mathcal{O}}(x,y_1) - \K(x,y_1)\mright)\mright)\cdot \poly(n,\ell) \tag{by \Cref{alg-A-prime}}\\ 
            &\leq \exp\mleft(O\mleft(\pK^{q(n), \mathcal{O}}(x) - \K(x,y_1)\mright)\mright)\cdot \poly(n,\ell) \tag{by \Cref{eq:good-witness-runtime}}\\
            &\leq  \exp\mleft(O\mleft(\pK^{q(n), \mathcal{O}}(x) - \K(x)\mright)\mright)\cdot \poly(n,\ell)
        \end{align*}
        It follows that the quantity in \Cref{eq:total-runtime} is at most
        \begin{align*}
            &\quad 2^{s'+1} \cdot \poly(t(n)) + 2^{s'+1}\cdot  \exp\mleft(O\mleft(\pK^{q(n), \mathcal{O}}(x) - \K(x)\mright)\mright)\cdot \poly(n,\ell) \tag{by $|\mathcal{S'}|\leq 2^{s'+1}$} \\
            &\leq \exp(s') \cdot  \exp\mleft(O\mleft(\pK^{q(n), \mathcal{O}}(x) - \K(x)\mright)\mright)\cdot \poly(n,\ell)\\
            &\leq \exp\mleft(\pK^{\rho^{-1}(t(n)), \mathcal{O}}(x) - \pK^{\rho(p(t(n))), \mathcal{O}}(x) + \log t(n) + \log \rho(p(t(n))) +\log p(t(n))\mright) \\
            &\quad\quad\cdot  \exp\mleft(O\mleft(\pK^{q(n), \mathcal{O}}(x) - \K(x)\mright)\mright)\cdot \poly(n,\ell)\tag{by \Cref{c:alg-depth}}\\
            &\leq \exp\mleft(\pK^{\rho^{-1}(t(n)), \mathcal{O}}(x) - \K(x) +O(\log n)\mright) \cdot  \exp\mleft(O\mleft(\pK^{q(n), \mathcal{O}}(x) - \K(x)\mright)\mright)\cdot \poly(n,\ell) \tag{by \Cref{l:K-pK}}\\
            &\leq \exp\mleft(\pK^{q(n), \mathcal{O}}(x) - \K(x) \mright) \cdot  \exp\mleft(O\mleft(\pK^{q(n), \mathcal{O}}(x) - \K(x)\mright)\mright)\cdot \poly(n,\ell)\\
            &\leq \exp\mleft(O\mleft(\pK^{q(n), \mathcal{O}}(x) - \K(x)\mright)\mright)\cdot \poly(n,\ell),
        \end{align*}
        where the second last inequality holds by letting $t(\cdot)\geq \rho(q(
        \cdot))$. 
        This completes the proof.
\end{proof}

\subsubsection{Proof of \texorpdfstring{\Cref{l:GapMINpnKT-to-PH}}{}}

We are now ready to show \Cref{l:GapMINpnKT-to-PH}.

\begin{proof}[Proof of \Cref{l:GapMINpnKT-to-PH}]
Assume $\mathsf{Gap}\textnormal{-}\mathsf{MINpKT}^{\mathsf{PH}}\in\prBPP$. 
Let $k\in\mathbb{N}$ and $L\in\Sigma^{\P}_k$.

Given \Cref{t:IL90}, it suffices to show how to decide $L$ on average over the uniform distribution. The idea is to obtain a heuristic scheme by using the worst-case algorithm from \Cref{l:alg-Sigma-k-cd} for deciding $L$, along with the fact that $\pK^{t,\mathsf{PH}}(x) - \K(x)$ is small for a uniformly random $x \in \bool^n$. Details follow.

Let $\mathcal{O}$ be a $\Sigma^{\P}_k$-complete language, and let $A$ be the algorithm from \Cref{l:alg-Sigma-k-cd} that satisfies the following. Given $\ell\in \mathbb{N}$ and $x\in\bool^n$, with probability at least $1-2^{-\ell}$, $A(x,1^{\ell})$ runs in time 
        \[
        2^{O\mleft(\pK^{q(n), \mathcal{O}}(x) - \K(x)\mright)}\cdot \poly(n,\ell)
        \]
        and decides whether $x\in L$, where $q$ is a polynomial satisfying that for all $x\in\bool^*$, $\pK^{q(|x|), \mathcal{O}}(x)\leq |x| + O(\log |x|)$. 

We now describe a heuristic scheme for solving $L$ over the uniform distribution. Let $B$ be the following algorithm:
\begin{quote}
    On input $(x,1^n,1^{\kappa})$, where $|x|=n$, set $\mathsf{cutoff}\vcentcolon=p(n,\kappa)$, where $p$ is a polynomial specified later. Run $A(x,1^{3})$ for $\mathsf{cutoff}$ steps. If it halts, output the returned answer; otherwise output $\bot$.
\end{quote}
    We argue the correctness of $B$. By the fact that for every $x\in\bool^n$, $A(x,1^{3})$ only errs with probability at most $2^{-3}=1/8$ (over the internal randomness of $A$), we get that
    \[
        \Prob_{B}\left[B(x, 1^n, 1^\kappa) \in \{L(x), \bot\}\right] \geq \frac{4}{5}.
    \]
    
    It remains to show that
        \[
    \Prob_{x \sim \bool^n}\left[\Prob_{B}[B(x, 1^n, 1^{\kappa}) = \bot] < \frac{1}{5}\right] \geq 1 - \frac{1}{\kappa}.
    \]
    Note that by construction, it is sufficient to show that with probability at least $1-1/\kappa$ over $x\sim\bool^n$, we have
    \[
\Prob_{A}\mleft[\mathsf{Runtime}\mleft(A(x,1^{3})\mright)\leq \mathsf{cutoff}\mright] > \frac{4}{5}.
    \]
    First of all, for every $x\in\bool^n$, we have that with probability at least $7/8$ (which is greater than $4/5$) over its internal randomness, $A(x,1^{3})$ runs in time
    \begin{equation}\label{eq:A-runtime}
        2^{O\mleft(\pK^{q(n), \mathcal{O}}(x) - \K(x)\mright)}\cdot \poly(n).
    \end{equation}
    Firstly, for every $x\in\bool^n$, we have 
    \[
        \pK^{q(n), \mathcal{O}}(x) \leq n +O(\log n).
    \]
    Secondly, by \Cref{l:incomp-strings}, with probability at least $1-1/\kappa$ over $x\sim\bool^n$,
    \[
    \K(x)\geq n-\log \kappa.
    \]
    It follows that with probability at least $1-1/\kappa$ over $x\sim\bool^n$, the runtime given by \Cref{eq:A-runtime} is at most
    \begin{align*}
        \exp\mleft(n+O(\log n)- n + \log \kappa\mright) \cdot \poly(n)\leq \poly(n,\kappa).
    \end{align*}
    By letting $p$ be a sufficiently large polynomial, the above is at most $\mathsf{cutoff}$. This completes the proof.
\end{proof}

\section{The Deterministic Case}\label{thm:deterministic}

\begin{theorem}\label{t:PH-heuristica-characterization-det}
	The following statements are equivalent\emph{:}
	\begin{enumerate}
		\item \emph{(Exclusion of PH-Heuristica).} $\Dist\mathsf{PH}\subseteq\Avg\P \implies \mathsf{PH}=\P$.
		\item \emph{(Meta-Complexity of  $\mathsf{K}^{\mathsf{PH}}$).} $\mathsf{Gap}\textnormal{-}\mathsf{MINKT}^{\mathsf{PH}} \in \P \implies \mathsf{Mild}\textnormal{-}\mathsf{Gap}\textnormal{-}\mathsf{MINKT}^{\mathsf{PH}} \in \P$.
	\end{enumerate}
\end{theorem}

As in the proof of \Cref{t:PH-heuristica-characterization}, \Cref{t:PH-heuristica-characterization-det} follows from the following two results.

\begin{theorem}[\cite{DBLP:conf/focs/Hirahara20}]\label{t:GapMINnKT-PH}
	The following holds:
	\[
         \Dist\mathsf{PH}\subseteq\Avg\P \iff
		\mathsf{Gap}\textnormal{-}\mathsf{MINKT}^{\mathsf{PH}} \in \P.
	\]
\end{theorem}

\begin{theorem}\label{t:MINnKT-NP}
	The following holds\emph{:}
	\[
		\NP=\P \iff \mathsf{Mild}\textnormal{-}\mathsf{Gap}\textnormal{-}\mathsf{MINKT}^{\mathsf{PH}} \in \P.
	\]
\end{theorem}

To prove \Cref{t:MINnKT-NP}, we first establish the following.

\begin{lemma}\label{t:MINnKT-NP-hard}
    If $\mathsf{Mild}\textnormal{-}\mathsf{Gap}\textnormal{-}\mathsf{MINKT}^{\mathsf{\SAT}} \in \BPP$, then $\NP \subseteq \BPP$.
\end{lemma}
\begin{proof}[Proof Sketch]
    The proof is essentially the same as that of \Cref{t:MINpnKT-NP-hard}. It suffices to consider the notion $\ell\lK^{t,\SAT}$ in place of $\ell\lpK^{t,\SAT}$. See also the proof of \cite[Theorem 34]{DBLP:journals/eccc/HuL025} for details.
\end{proof}

\begin{proof}[Proof of \Cref{t:MINnKT-NP}]
The forward direction is immediate since $\mathsf{Mild}\textnormal{-}\mathsf{Gap}\textnormal{-}\mathsf{MINKT}^{\mathsf{PH}}$ lies in $\PH$. Note that $\PH=\P$ if $\NP=\P$.

For the backward direction, first note that \Cref{t:MINnKT-NP-hard} implies the following.
\[
\mathsf{Mild}\textnormal{-}\mathsf{Gap}\textnormal{-}\mathsf{MINKT}^{\mathsf{\PH}} \in \P \implies \NP\subseteq\BPP.
\]
Moreover, as mentioned in \Cref{ss:intro_techniques}, \cite{DBLP:conf/focs/Hirahara20} shows that if $\mathsf{Gap}\textnormal{-}\mathsf{MINKT}^{\mathsf{\PH}} \in \P$, then $\BPP = \P$ (see \cite[Theorem 1.17]{DBLP:conf/focs/Hirahara20}). It follows that if $\mathsf{Mild}\textnormal{-}\mathsf{Gap}\textnormal{-}\mathsf{MINKT}^{\mathsf{\PH}} \in \P$, then we have both $\NP \subseteq \BPP$ and $\BPP = \P$, which together imply $\NP = \P$, as desired.
\end{proof}

\section{Equivalence of Gap and No-Gap Hardness in the Average-Case Setting}\label{sec:average}

In this section, we show that, in the average-case setting (both zero-error and two-sided-error), the hardness of the standard gap version of $\mathsf{MINpKT}^{\mathsf{PH}}$ is equivalent to that of its no-gap variant.

\subsection{The Zero-Error Case}
We first prove the zero-error case.
\begin{theorem}\label{t:gap-mild-gap-avg-errless}
The following holds:
\[
\textnormal{``}\mathsf{Gap}\textnormal{-}\mathsf{MINpKT}^{\mathsf{PH}}\in \mathsf{AvgBPP}\textnormal{''}
\iff
\textnormal{``}\mathsf{MINpKT}^{\mathsf{PH}}\in \mathsf{AvgBPP}\textnormal{''}.
\]
\end{theorem}
    We will need the following worst-case to average-case reduction result for approximating $\pK^t$.

    \begin{theorem}[Following \cite{DBLP:conf/focs/Hirahara18, GKLO22}]\label{t:worst-to-avg-red}
    Let $\mathcal{O}$ be any oracle. If $\textnormal{``}\mathsf{Gap}\textnormal{-}\mathsf{MINpKT}^{\mathcal{O}}\in\mathsf{AvgBPP}\textnormal{''}$, then $
\mathsf{Gap}\textnormal{-}\mathsf{MINpKT}^{\mathcal{O}}\in\mathsf{prBPP}$.
    \end{theorem}

We are now ready to show \Cref{t:gap-mild-gap-avg-errless}.
\begin{proof}[Proof of \Cref{t:gap-mild-gap-avg-errless}]
The reverse implication is immediate. We therefore prove the forward implication.

Assume that
\[
\textnormal{``}\mathsf{Gap}\textnormal{-}\mathsf{MINpKT}^{\mathsf{PH}}\in \mathsf{AvgBPP}\textnormal{''}.
\]
By \Cref{t:worst-to-avg-red}, this implies
\[
\mathsf{Gap}\textnormal{-}\mathsf{MINpKT}^{\mathsf{PH}}\in \mathsf{prBPP}.
\]
Then, by \Cref{l:GapMINpnKT-to-PH}, we obtain
\[
\Dist\mathsf{PH}\subseteq \mathsf{AvgBPP}.
\]

Consequently, it suffices to establish the following claim.

\begin{claim}\label{c:minpkt-avg-easy}
If $\Dist\mathsf{PH}\subseteq \mathsf{AvgBPP}$, then ``$\mathsf{MINpKT}^{\mathsf{PH}} \in \mathsf{AvgBPP}$''.
\end{claim}

\begin{proof}[Proof of \Cref{c:minpkt-avg-easy}]
\renewcommand{\qedsymbol}{$\diamond$}

Assume $\Dist\mathsf{PH}\subseteq \mathsf{AvgBPP}$. Fix an oracle $\mathcal{O}\in \mathsf{PH}$ and let $\{\mathcal{D}_n\}$ be any polynomial-time samplable distribution family. We show that $\mathsf{MINpKT}^{\mathcal{O}}$ is solvable with zero error on average with respect to $\{\mathcal{D}_n\}$.

First, we observe that $\mathsf{MINpKT}^{\mathcal{O}}\in \promise\AM^{\mathsf{PH}}$. Consider the following protocol on input $(x,1^s,1^t,1^a,1^b)$:

\begin{quote}
The verifier samples a random string $w \in \{0,1\}^t$ and sends it to the prover. The prover responds with an oracle program $M_w^{(\cdot)}$. The verifier then simulates $M_w^{\mathcal{O}}(w)$ for at most $t$ steps and accepts iff
\[
|M_w| \le s \quad \text{and} \quad M_w^{\mathcal{O}}(w) = x.
\]
\end{quote}

If $\pK^{t,\mathcal{O}}_{b/a}(x) \le s$, then there exists a strategy for the prover that causes the verifier to accept with probability at least $b/a$. Conversely, if $\pK^{t,\mathcal{O}}_{(b-1)/a}(x) > s$, then every prover causes acceptance with probability less than $(b-1)/a$. Standard error reduction yields a $\promise \mathsf{AM}^{\mathsf{PH}}$ protocol for $\mathsf{MINpKT}^{\mathcal{O}}$.

Since $\promise\AM^{\mathsf{PH}} \subseteq \promise\mathsf{PH}$, it follows that $\mathsf{MINpKT}^{\mathcal{O}} \in \promise\mathsf{PH}$. Hence there exists a $\mathsf{PH}$ predicate $B$ that correctly decides $\mathsf{MINpKT}^{\mathcal{O}}$ on the promise inputs.

Now consider the family of distributions $\mathcal{D}'_{\langle n,s,t,a,b\rangle}$, where for each each triple $\langle n,s,t, a, b\rangle$, we sample $x \sim \mathcal{D}_n$ and output $(x,1^s,1^t,1^a,1^b)$.

By the assumption $\Dist\mathsf{PH}\subseteq \mathsf{AvgBPP}$, there exists a zero-error heuristic scheme $A$ that solves $B$ under $\{\mathcal{D}'_{\langle n,s,t,a,b\rangle}\}$.
It is easy to see that $A$ is also a zero-error heuristic scheme for $\mathsf{MINpKT}^{\mathcal{O}}$.
\end{proof}

The theorem follows immediately from the claim.
\end{proof}

\subsection{The Two-Sided-Error Case}

Here we show that the hardness of the standard gap version of $\mathsf{MINpKT}^{\mathsf{PH}}$ is equivalent to that of its no-gap variant in the two-sided-error average-case setting.
\begin{theorem}\label{t:gap-mild-gap-avg-err}
The following holds:
\[
\textnormal{``}\mathsf{Gap}\textnormal{-}\mathsf{MINpKT}^{\mathsf{PH}}\in \mathsf{HeurBPP}\textnormal{''}
\iff
\textnormal{``}\mathsf{MINpKT}^{\mathsf{PH}}\in \mathsf{HeurBPP}\textnormal{''}.
\]
\end{theorem}

We will need the following result.
\begin{theorem}[{\cite{DBLP:conf/innovations/HiraharaS22a}}]\label{t:ph-pessilang}
The following statements are equivalent.\footnote{This formulation differs slightly from the original statement in \cite{DBLP:conf/innovations/HiraharaS22a}. More precisely, here we use the same language $f$ for every choice of $c$.  It is not hard to  verify that this version of the statement follows from the proofs in that paper.}
\begin{enumerate}
    \item $\Dist\PH \not\subseteq \mathsf{HeurBPP}$.

    \item\label{ph-prg}
    There exists a language $f \in \PH$ and a constant $d>0$ such that, for every constant $c>0$, there exists a pseudorandom generator
    \[
        G=\{G_n:\bool^n\to\bool^{\,n+c\log n}\}_{n\in\mathbb{N}}
    \]
    satisfying the following conditions:
    \begin{itemize}
        \item $G$ is infinitely often secure against probabilistic polynomial-time algorithms. That is, for every probabilistic polynomial-time algorithm $A$ and every polynomial $p$, there exist infinitely many $n$ such that
    \[
        \mleft|
        \Prob_{r\sim\bool^n}\!\left[A(G_n(r))=1\right]
        -
        \Prob_{x\sim\bool^{\,n+c\log n}}\!\left[A(x)=1\right]
        \mright|
        \le \frac{1}{p(n)}.
    \]
    
        \item  For all sufficiently large $n$, $G_n$ is computable in time $n^d$ with oracle access to $f$.
        
        \item  $G$ is \emph{min-entropy-preserving} in the sense that for every $n\in\mathbb{N}$ and $y\in \mathsf{Support}(G_n)$, $\Prob_{r\sim\bool^n}[G_n(r)=y]=2^{-n}$, i.e., $ H_\infty(G_n(\mathcal{U}_n))=n$.\footnote{In fact, the PRG is based on the direct product generator in which the seed is part of the output; in particular, the PRG is seed-extending.}
    \end{itemize}

\end{enumerate}
\end{theorem}

We are now ready to prove \Cref{t:gap-mild-gap-avg-err}.
\begin{proof}[Proof of \Cref{t:gap-mild-gap-avg-err}]
Since the backward direction is immediate, we proceed to prove the forward direction.

Assume that $\textnormal{``}\mathsf{Gap}\textnormal{-}\mathsf{MINpKT}^{\mathsf{PH}} \in \mathsf{HeurBPP}\textnormal{''}$.
We first show that infinitely-often secure $\mathsf{PH}$-computable pseudorandom generators with arbitrary logarithmic stretch do not exist.

\begin{claim}\label{t:prg-pK}
If $\mathsf{Gap}\textnormal{-}\mathsf{MINpKT}^{\mathsf{PH}} \in \mathsf{HeurBPP}$, then \Cref{ph-prg} in \Cref{t:ph-pessilang} is false.
\end{claim}

\begin{proof}[Proof of \Cref{t:prg-pK}]
\renewcommand{\qedsymbol}{$\diamond$}
In the following, we always consider the $\mathsf{Gap}\textnormal{-}\mathsf{MINpKT}$ problem with input parameters $a=3$ and $b=2$, and omit these parameters from the notation.

For the sake of contradiction, suppose that both ``$\mathsf{Gap}\textnormal{-}\mathsf{MINpKT}^{\mathsf{PH}} \in \mathsf{HeurBPP}$'' and \Cref{ph-prg} of \Cref{t:ph-pessilang} hold. Let
\begin{itemize}
    \item $f\in\PH$ be the language guaranteed by \Cref{ph-prg} of \Cref{t:ph-pessilang};
    \item $\rho$ be a polynomial such that $\mathsf{Gap}_{\rho}\textnormal{-}\mathsf{MINpKT}^{f}\in\mathsf{HeurBPP}$;
    \item $\alpha>\rho$ be a polynomial specified later, and let
    \[
    G=\{G_n\colon\{0,1\}^n\to\{0,1\}^{m(n)}\}
    \]
    be the pseudorandom generator guaranteed by \Cref{ph-prg} of \Cref{t:ph-pessilang}, where
    \[
    m(n)\coloneqq n+\log\alpha(n),
    \]
    and for all sufficiently large $n$, $G_n$ is computable in time $t(n)\coloneqq n^d$
    for some constant $d>0$ that is independent of $\alpha$;
    \item $\mathcal{D}=\{D_{m(n)}\}$ where each $D_{m(n)}$ samples from $G_n(\mathcal{U}_n)$.
\end{itemize}

Next, starting from an algorithm witnessing the easiness of $\mathsf{Gap}_{\rho}\textnormal{-}\mathsf{MINpKT}^{f}$ and applying standard error-reduction techniques, we obtain a probabilistic polynomial-time algorithm $A$ such that\footnote{We fix the probability threshold $b/a$ to $2/3$ and omit these input parameters in the subsequent discussion.}
\begin{equation}\label{eq:correct}
\Prob_{x\sim\bool^{m(n)}, A}\!\left[
A \textnormal{ decides }
\mathsf{Gap}_{\rho}\textnormal{-}\mathsf{MINpKT}^{f}
\textnormal{ on }(x,1^{s(n)},1^{t(n)})
\right]
\ge
1-\frac1{\alpha(n)^2},
\end{equation}
where $s(n)\vcentcolon=n+\gamma\log n$ for some absolute constant $\gamma>0$ specified later, and $t(n)$ is specified as above.

We will construct from $A$ an efficient algorithm that distinguishes $G_n(\mathcal{U}_n)$ from $\mathcal{U}_{m(n)}$ for every sufficiently large $n$. Fix such an $n$ for the remainder of the proof.

First, note that for every $y\in\mathsf{Support}(D_{m(n)})$, we have $D_{m(n)}(y)\leq 2^{-{n}}=2^{-m}\cdot \alpha(n)$. In other words, $\mathcal{U}_{m(n)}$ dominates $D_{m(n)}$, and therefore

\begin{align*}
    &\quad\Prob_{y\sim D_{m(n)}, A}\!\left[
A \textnormal{ fails to decide }
\mathsf{Gap}_{\rho}\textnormal{-}\mathsf{MINpKT}^{f}
\textnormal{ on }(y,1^{s(n)},1^{t(n)})
\right]\\
&\leq \Prob_{x\sim\mathcal{U}_{m(n)}, A}\!\left[
A \textnormal{ fails to decide }
\mathsf{Gap}_{\rho}\textnormal{-}\mathsf{MINpKT}^{f}
\textnormal{ on }(x,1^{s(n)},1^{t(n)})
\right]\cdot \alpha(n)\\
&\leq \frac{1}{\alpha(n)^2}\cdot \alpha(n) \tag{by \Cref{eq:correct}}\\
&= \frac{1}{\alpha(n)}.
\numberthis\label{eq:distinguish-1}
\end{align*}
Also, we have for every $r\in\bool^n$,
\[
\pK^{t(n),f}(G_n(r))
\le
n+O(\log n)\leq s(n).
\]
Hence every such $(G_n(r), 1^{s(n)},1^{t(n)})$ is a YES instance of
$\mathsf{Gap}_{\rho}\textnormal{-}\mathsf{MINpKT}^{f}$. It follows that
\begin{align*}
    1-\frac{1}{\alpha(n)} &\leq \Prob_{y\sim D_{m(n)}, A}\!\left[
A \textnormal{ decides }
\mathsf{Gap}_{\rho}\textnormal{-}\mathsf{MINpKT}^{f}
\textnormal{ on }(y,1^{s(n)},1^{t(n)})
\right] \tag{by \Cref{eq:distinguish-1}}\\
& =     \Prob_{r\sim\bool^{n}, A}\!\left[
A(G_n(r),1^{s(n)},1^{t(n))})=1
\right].
\numberthis\label{eq:distinguish-2}
\end{align*}

On the other hand, by \Cref{l:incomp-strings}, with probability except $1/n$ over $x\sim\mathcal{U}_{m(n)}$,
\[
\pK_{1/3}^{\rho(t(n)+5), f}(x)
\geq 
m(n)-O(\log n)
>
s(n)+\log\rho(t(n)+5),
\]
where the final inequality holds provided that $\alpha$ is chosen to be a sufficiently large polynomial. 
That is, with probability except $1/n$ over $x\sim\mathcal{U}_{m(n)}$, $(x, 1^{s(n)}, 1^{t(n)})$ is a NO instance of $\mathsf{Gap}_{\rho}\textnormal{-}\mathsf{MINpKT}^{f}$. 
It follows from the correctness of $A$ that
\begin{equation}\label{eq:distinguish-3}
\Prob_{x \sim \mathcal{U}_{m(n)}, A}\Big[
A(x, 1^{s(n)}, 1^{t(n)})=0\Big] \geq  1 - \frac{1}{n}-\frac{1}{\alpha(n)^2}.
\end{equation}

Comparing \Cref{eq:distinguish-2,eq:distinguish-3}, we obtain a probabilistic polynomial-time algorithm that distinguishes $G_n(\mathcal{U}_n)$ from $\mathcal{U}_{m(n)}$ for every sufficiently large $n$, a contradiction.
\end{proof}

It follows from \Cref{t:prg-pK} and \Cref{t:ph-pessilang} that $\DistPH\subseteq\mathsf{HeurBPP}$. Thus, it suffices to establish the following claim.

\begin{claim}\label{c:minpkt-heur-easy}
If $\DistPH\subseteq\mathsf{HeurBPP}$, then $
\textnormal{``}\mathsf{MINpKT}^{\mathsf{PH}}\in\mathsf{HeurBPP}\textnormal{''}$.
\end{claim}

\begin{proof}[Proof of \Cref{c:minpkt-heur-easy}]
\renewcommand{\qedsymbol}{$\diamond$}
The proof is essentially identical to that of \Cref{c:minpkt-avg-easy}.
\end{proof}

The theorem now follows immediately from \Cref{c:minpkt-heur-easy} and the preceding discussion.
\end{proof}

\section{A Relativization Barrier}\label{ss:ave-case-rel-barrier}

In this section, we exhibit an oracle world where  $\mathsf{Gap}\textnormal{-}\mathsf{MINpKT}^{\mathsf{PH}}$ is easy while $\mathsf{Mild}\textnormal{-} \mathsf{Gap}\textnormal{-}\mathsf{MINpKT}^{\mathsf{PH}}$ is hard. 
To make the statement that we aim to prove precise, we first explain the relativization of these computational problems. As usual, this is obtained by providing oracle access to all machines involved in the computation.

\paragraph{Relativization.} For a string $x \in \{0,1\}^n$, a time bound $t$, a parameter $\lambda \in [0,1]$, and languages $L, \mathcal{O} \subseteq \{0,1\}^*$, the $(L,\mathcal{O})$\emph{-oracle  probabilistic $t$-time bounded Kolmogorov complexity} of $x$ is defines as 
   \[
        \pK^{t, L, \mathcal{O}}_\lambda(x) \eqdef \min\mleft\{s\in\mathbb{N}\;\middle\vert\;\Prob_{r\sim\{0,1\}^t}\bigg [ \K^{t, L, \mathcal{O}}(x \mid r)\le s \bigg ] \ge \lambda \mright\},
    \]
where $\K^{t, L, \mathcal{O}}(x \mid r)$ denotes the conditional $t$-time bounded Kolmogorov complexity of $x$ with oracle access to both $L$ and $\mathcal{O}$.

For a function $\rho\colon\mathbb{N}\to\mathbb{N}$ and oracles $L$ and $\mathcal{O}$, let $\mathsf{Gap}_{\rho}$-$\mathsf{MINpKT}^{L, \mathcal{O}}$ be the following promise problem: Given $(x,1^s,1^t,1^a,1^b)$, where $x\in \bool^*$ and $s,t,a,b \in \mathbb{N}$, decide whether $\pK^{t, L, \mathcal{O}}_{b/a}(x) \leq s$ (``yes'' case) or $\pK^{\rho(t+a+b), L, \mathcal{O}}_{(b-1)/a}(x) \geq s+\log\rho(t+a+b)$ (``no'' case). 

Let $\mathcal{O}$ be an oracle. We say that $\mathsf{Gap}\textnormal{-}\mathsf{MINpKT}^{\mathsf{PH^{\mathcal{O}}},\mathcal{O}} \in \prBPP^{\mathcal{O}}$ if for every language $L \in \PH^\mathcal{O}$, there is a polynomial $\rho$ and an oracle probabilistic polynomial-time algorithm that solves $\mathsf{Gap}_{\rho}\textnormal{-}\mathsf{MINpKT}^{L, \mathcal{O}}$ when given oracle access to $\mathcal{O}$. 

Analogously,  we consider the collection of problems $\mathsf{Mild}\textnormal{-} \mathsf{Gap}\textnormal{-}\mathsf{MINpKT}^{\mathsf{PH^{\mathcal{O}},\mathcal{O}}}$, for each $L \in \PH^\mathcal{O}$. Each problem $\mathsf{Mild}\text{-}\mathsf{Gap}_{\rho}$-$\mathsf{MINpKT}^{L, \mathcal{O}}$ is similar to $\mathsf{Gap}_{\rho}\textnormal{-}\mathsf{MINpKT}^{L, \mathcal{O}}$ except for that in the no case, we have $\pK^{t+\rho(|x|+a+b), L, \mathcal{O}}_{(b-1)/a}(x) \geq s+\log\rho(t+a+b)$ instead of $\pK^{\rho(t+a+b), L, \mathcal{O}}_{(b-1)/a}(x) \geq s+\log\rho(t+a+b)$. 

The statement $\mathsf{Mild}\textnormal{-} \mathsf{Gap}\textnormal{-}\mathsf{MINpKT}^{\mathsf{PH^{\mathcal{O}},\mathcal{O}}} \not\in \prBPP^{\mathcal{O}}$ means that there is a language $L \in \PH^\mathcal{O}$ such that for no constant $c \geq 1$ and polynomial $\rho(n) = c \cdot n^c$ we have $\mathsf{Mild}\text{-}\mathsf{Gap}_{\rho}$-$\mathsf{MINpKT}^{L, \mathcal{O}} \in \prBPP^\mathcal{O}$.

\begin{theorem}[Oracle world where  $\mathsf{Gap}\textnormal{-}\mathsf{MINpKT}^{\mathsf{PH}}$ is easy while $\mathsf{Mild}\textnormal{-} \mathsf{Gap}\textnormal{-}\mathsf{MINpKT}^{\mathsf{PH}}$ is hard]\label{thm:intro_oracle_separation}~\\
    There is an oracle $\mathcal{O}$ such that $\mathsf{Gap}\textnormal{-}\mathsf{MINpKT}^{\mathsf{PH^{\mathcal{O}}},\mathcal{O}} \in \prBPP^{\mathcal{O}}$ but $\mathsf{Mild}\textnormal{-} \mathsf{Gap}\textnormal{-}\mathsf{MINpKT}^{\mathsf{PH^{\mathcal{O}},\mathcal{O}}} \not\in \prBPP^{\mathcal{O}}$. 
\end{theorem}

Recall that, for an oracle $\mathcal{O}$, $\mathsf{DistPH}^\mathcal{O}$ denotes the set of all distributional problems $(L,\mathcal{D})$ where $L \in \PH^\mathcal{O}$ and $\mathcal{D} \in \mathsf{PSAMP}^\mathcal{O}$. The proof of \Cref{thm:intro_oracle_separation} relies on the following oracle separation, which is an immediate consequence of  \citep{DBLP:conf/coco/Impagliazzo11, DBLP:conf/focs/HiraharaN21}. 
 
\begin{theorem}[\citep{DBLP:conf/coco/Impagliazzo11, DBLP:conf/focs/HiraharaN21}]\label{thm:oracle_sep}
    There is an oracle $\mathcal{O} \subseteq \{0,1\}^*$ such that $\mathsf{DistPH}^\mathcal{O} \subseteq \mathsf{AvgBPP}^\mathcal{O}$ and $\PH^\mathcal{O} \nsubseteq \mathsf{BPP}^\mathcal{O}$.
\end{theorem}

We are now ready to prove \Cref{thm:intro_oracle_separation}.

\begin{proof}[Proof Sketch of \Cref{thm:intro_oracle_separation}.]
    Since the argument is straightforward given \Cref{thm:oracle_sep} and the proof of \Cref{t:PH-heuristica-characterization}, we only provide a sketch here. We consider the oracle $\mathcal{O} \subseteq \{0,1\}^*$  provided by \Cref{thm:oracle_sep}, and verify that it satisfies the required properties in \Cref{thm:intro_oracle_separation}.

\paragraph{1) Proof that $\mathsf{Gap}\textnormal{-}\mathsf{MINpKT}^{\mathsf{PH^{\mathcal{O}}},\mathcal{O}} \in \prBPP^{\mathcal{O}}$.} By assumption, $\mathsf{DistPH}^\mathcal{O} \subseteq \mathsf{AvgBPP}^\mathcal{O}$. Let $L \in \PH^\mathcal{O}$. We need to prove that there is a polynomial $\rho$ and an oracle probabilistic polynomial-time algorithm that solves $\mathsf{Gap}_{\rho}\textnormal{-}\mathsf{MINpKT}^{L, \mathcal{O}}$ when given oracle access to $\mathcal{O}$. We rely on the following lemma.\footnote{Analogously to the definitions introduced above, the measure $\K^{t, L_1, L_2}(x)$ and the corresponding meta-computational problem $\mathsf{MINKT}^{L_1,L_2}$ are defined in the natural way, i.e., by considering programs that have oracle access to both $L_1$ and $L_2$.}

\begin{lemma}[Relativized form of \Cref{worst-to-avg-Kt-PH}]\label{lem:rel-worst-to-avg-Kt-PH}
    For all oracles $L_1, L_2 \subseteq \{0,1\}^*$, if $(\mathsf{MINKT}^{L_1,L_2},\PSAMP) \in  \mathsf{AvgBPP}^{L_2}$, then $\mathsf{Gap}_\rho\textnormal{-}\mathsf{MINpKT}^{L_1, L_2}\in\prBPP^{L_2}$ for some polynomial $\rho$. 
\end{lemma}

 Assuming \Cref{lem:rel-worst-to-avg-Kt-PH}, we proceed as follows. We apply the lemma with $L_1 = L$ and $L_2 = \mathcal{O}$, which simplifies the assumption to $(\mathsf{MINKT}^{L,\mathcal{O}},\mathsf{PSAMP}^{\mathcal{O}}) \subseteq \mathsf{AvgBPP}^{\mathcal{O}}$. Since it is not hard to see that  $\mathsf{MINKT}^{L, \mathcal{O}} \in \PH^\mathcal{O}$, the assumption of \Cref{lem:rel-worst-to-avg-Kt-PH} holds (using that $\mathsf{DistPH}^\mathcal{O} \subseteq \mathsf{AvgBPP}^\mathcal{O}$). Therefore, it follows that  $\mathsf{Gap}_\rho\textnormal{-}\mathsf{MINpKT}^{L, \mathcal{O}}\in\prBPP^{\mathcal{O}}$ for some polynomial $\rho$, as desired. 

It remains to prove \Cref{lem:rel-worst-to-avg-Kt-PH}. The proof is a simple adaptation of that of {\cite[Proposition 58]{DBLP:journals/eccc/HuL025}} to the relativized setting. Since it requires no new idea, we omit the details.
    
    \paragraph{2) Proof that $\mathsf{Mild}\textnormal{-} \mathsf{Gap}\textnormal{-}\mathsf{MINpKT}^{\mathsf{PH^{\mathcal{O}},\mathcal{O}}} \not\in \prBPP^{\mathcal{O}}$.}
    Recall that, by construction, $\PH^\mathcal{O} \nsubseteq \mathsf{BPP}^\mathcal{O}$. Next, we use this separation to argue that $\mathsf{Mild}\textnormal{-} \mathsf{Gap}\textnormal{-}\mathsf{MINpKT}^{\mathsf{PH^{\mathcal{O}},\mathcal{O}}} \not\in \prBPP^{\mathcal{O}}$. In more detail, we establish the contrapositive, i.e., if $\mathsf{Mild}\textnormal{-} \mathsf{Gap}\textnormal{-}\mathsf{MINpKT}^{\mathsf{PH^{\mathcal{O}},\mathcal{O}}} \in \prBPP^{\mathcal{O}}$ then $\PH^\mathcal{O} \subseteq \mathsf{BPP}^\mathcal{O}$. Since $\NP^\mathcal{O} \subseteq \BPP^\mathcal{O}$ yields $\PH^\mathcal{O} \subseteq \BPP^\mathcal{O}$, it suffices to argue the assumption implies that $\NP^\mathcal{O} \subseteq \BPP^\mathcal{O}$.

    To achieve this, it is sufficient to establish the following lemma.

\begin{lemma}[Relativized form of \Cref{t:MINpnKT-NP-hard}]\label{lem:ext_of_MINpnKT-NP-hard}
For every $\mathcal{O} \subseteq \{0,1\}^*$, if
      $\mathsf{Mild}\textnormal{-} \mathsf{Gap}\textnormal{-}\mathsf{MINpKT}^{\mathsf{\PH^\mathcal{O}, \mathcal{O}}} \in \prBPP^\mathcal{O}$, then $\NP^\mathcal{O}\subseteq\BPP^\mathcal{O}$.
\end{lemma}

It is straightforward to check that the proof of \Cref{t:MINpnKT-NP-hard} extends to \Cref{lem:ext_of_MINpnKT-NP-hard}. The key point is that all techniques employed in the proof relativize. We omit the details.
\end{proof}

\bibliographystyle{alpha}
\bibliography{references}

\end{document}